\documentclass{article}
\usepackage{booktabs}
\usepackage{float}
\usepackage{subcaption}
\usepackage{graphicx}
\usepackage{calc}
\usepackage{amsfonts}
\usepackage{amssymb}
\usepackage{mathtools}
\usepackage{dsfont}
\usepackage{algorithm}
\usepackage[noend]{algpseudocode}
\usepackage{siunitx}
\usepackage{amsthm}
\usepackage{amsmath}
\usepackage{hyperref}
\usepackage{authblk}
\usepackage{geometry}

\newcommand{\NN}{\mathbb{N}}
\newcommand{\RR}{\mathbb{R}}
\newcommand{\CC}{\mathbb{C}}
\newcommand{\PP}{\mathbb{P}}
\newcommand{\EE}{\mathbb{E}}
\DeclareMathOperator{\SO}{SO}
\algrenewcommand\algorithmicrequire{\textbf{Input:}}
\algrenewcommand\algorithmicensure{\textbf{Output:}}
\DeclareMathOperator*{\argmin}{arg\,min}
\DeclareMathOperator*{\argmax}{arg\,max}
\DeclareMathOperator*{\st}{\mathrm{St}}
\renewcommand{\Re}{\operatorname{Re}}

\newcommand{\lp}{(}\newcommand{\rp}{)}

\newcommand\CONDITION[2]%
  {\begin{tabular}[t]{@{}l@{}l@{}}
     #1&#2
   \end{tabular}%
  }
\algdef{SE}[WHILE]{While}{EndWhile}[1]%
  {\algorithmicwhile\ \CONDITION{#1}{\ \algorithmicdo}}%
  {\algorithmicend\ \algorithmicwhile}
\algnewcommand{\IfThenElse}[3]{
  \State \algorithmicif\ #1\ \algorithmicthen\ #2\ \algorithmicelse\ #3}
  \algnewcommand{\IfThen}[3]{
  \State \algorithmicif\ #1\ \algorithmicthen\ #2}

\theoremstyle{plain}
\newtheorem{theorem}{Theorem}
\newtheorem{lemma}{Lemma}
\newtheorem{corollary}{Corollary}
\newtheorem{example}{Example}
\newtheorem{remark}{Remark}
\theoremstyle{definition}
\newtheorem{definition}{Definition} 

\begin{document}

\title{Stereological determination of particle size distributions for similar convex bodies}

\author[1]{Thomas van der Jagt}
\author[1]{Geurt Jongbloed}
\author[2]{Martina Vittorietti}
\affil[1]{Delft Institute of Applied Mathematics, Delft University of Technology.}
\affil[2]{Scienze Economiche, Aziendali e Statistiche, Universit\`a degli studi di Palermo.}
\date{}
\maketitle

\begin{abstract}
Consider an opaque medium which contains 3D particles. All particles are convex bodies of the same shape, but they vary in size. The particles are randomly positioned and oriented within the medium and cannot be observed directly. Taking a planar section of the medium we obtain a sample of observed 2D section profile areas of the intersected particles. In this paper the distribution of interest is the underlying 3D particle size distribution for which an identifiability result is obtained. Moreover, a nonparametric estimator is proposed for this size distribution. The estimator is proven to be consistent and its performance is assessed in a simulation study.
\end{abstract}

\section{Introduction}
In the classical Wicksell corpuscle problem \cite{Wicksell1925} spheres are randomly positioned in an opaque body. The problem is to estimate the size distribution of the spheres using the circular profiles observed in a planar section. The motivation of the problem originated from anatomy as well as astronomy. In the anatomical setting it is of interest as so-called follicles may be observed in slices of organs during post-mortem studies. Such follicles are approximately spherical, resulting in approximately circular section profiles from the intersected follicles. We may then wonder what is the distribution of the radii of the follicles. Questions of similar nature appear in the field of materials science. An important feature of the so-called microstructure of a steel are the grains. Knowing the size distribution of the 3D grains allows for studying the relationship between grain size distribution and mechanical properties of the metal. It is much simpler to obtain 2D information by observing a planar cross section of the metal, compared to obtaining 3D information of the steel's microstructure. Hence, it is of interest to use the 2D observations for estimating 3D information. These problems belong to the field of stereology, which deals with the estimation of higher dimensional information from lower dimensional samples. 

\begin{figure}[t]%
    \centering
    \makebox[\linewidth]{\makebox[\linewidth]{
    \begin{subfigure}[t]{0.65\linewidth}
        \centering
        \includegraphics[height=4.4cm]{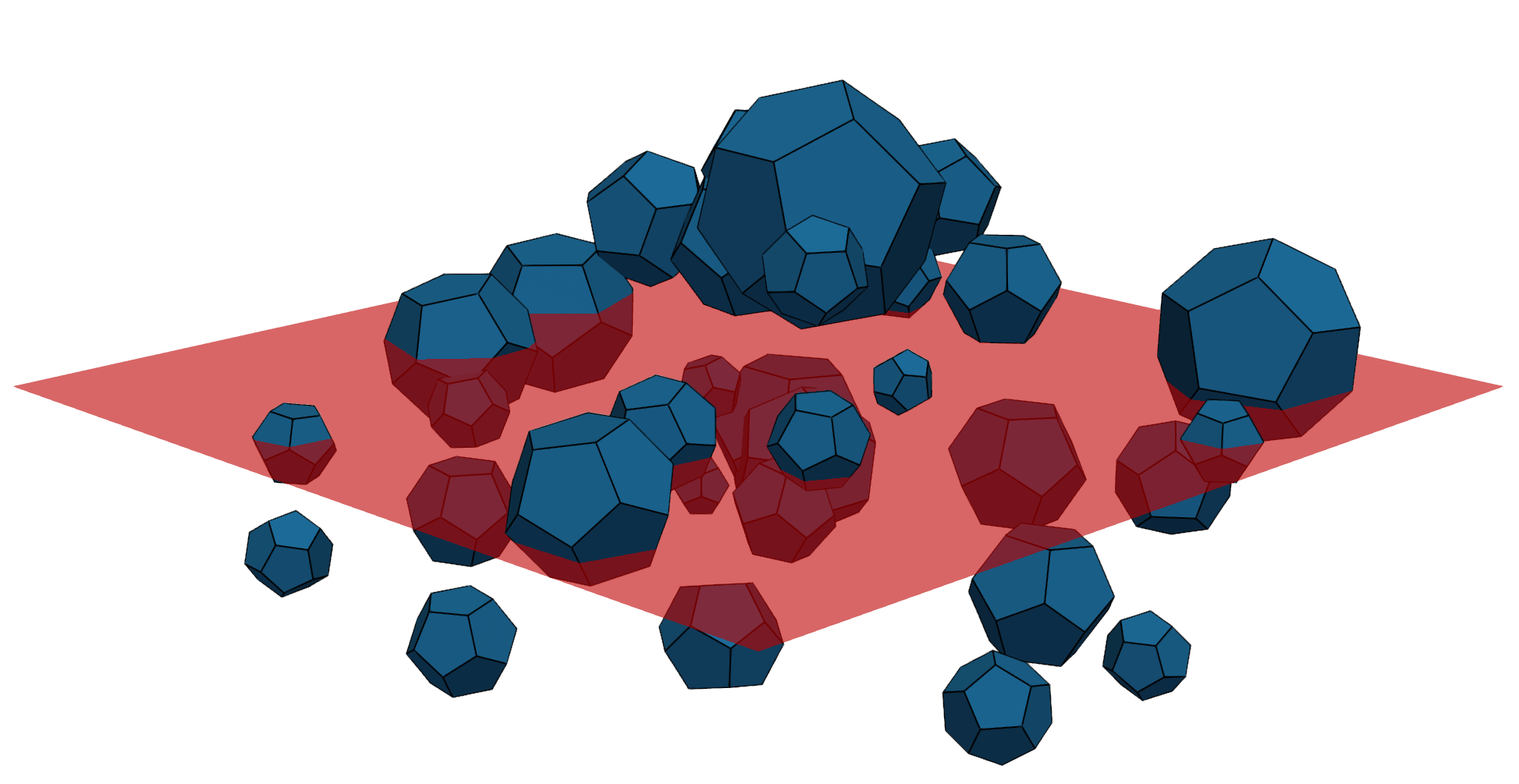}
    \end{subfigure}
    \begin{subfigure}[t]{0.35\linewidth}
        \centering
        \includegraphics[height=4.4cm]{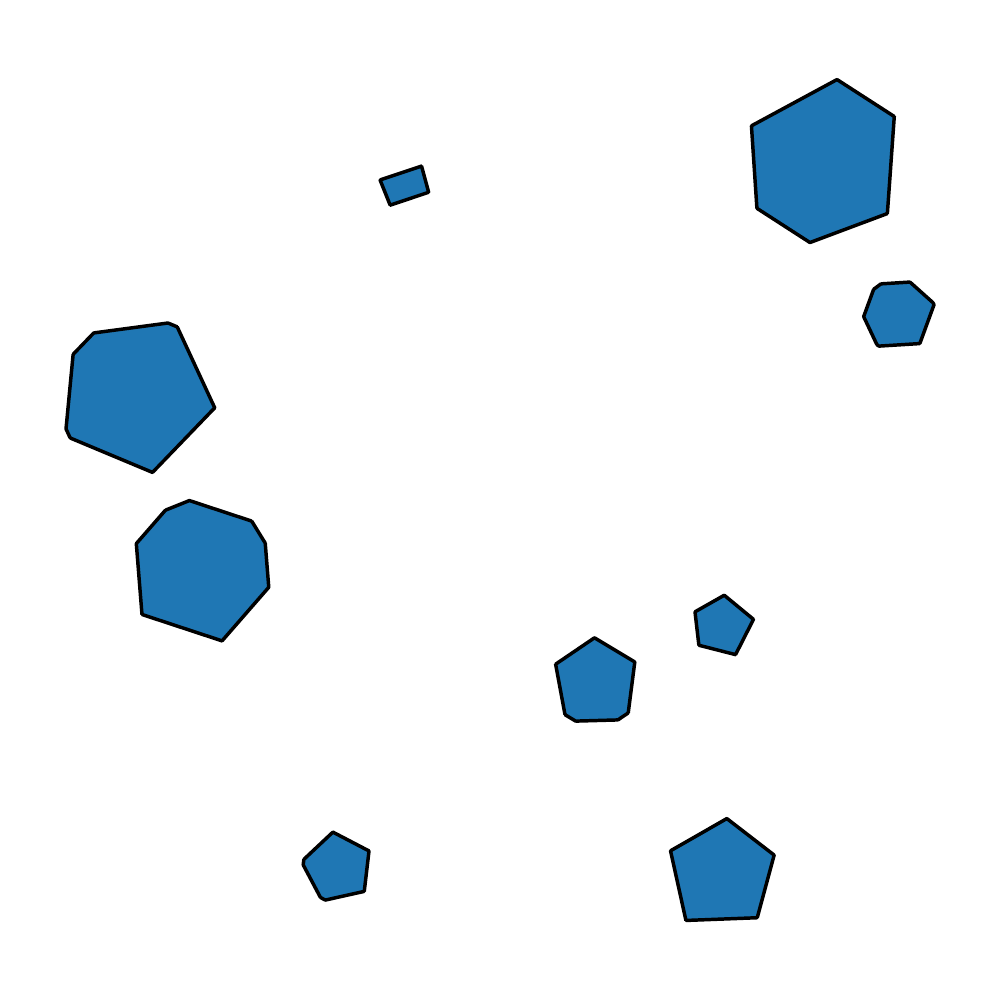}
    \end{subfigure}\hfill}}
    \caption{Left: Random spatial system of convex dodecahedra intersected with a plane. Right: Observed section profiles.}\label{problem_visualization}%
\end{figure}

We study a generalization of the Wicksell problem. Consider 3D particles, convex bodies to be precise, which are randomly positioned in an opaque body and randomly oriented. A convex body is a compact and convex set with non-empty interior. These particles all have the same known shape, but they do not have the same size. The particles cannot be observed directly, we only observe 2D section profiles of these particles in a planar section. We address the statistical problem of estimating the size distribution of the particles, using a sample of observed areas of the section profiles. A visualization of the problem setting is given in Figure \ref{problem_visualization}. In this particular example each particle is a convex dodecahedron.

An overview of estimators for the size distribution in the spherical setting is presented in \cite{Chiu2013}. The problem has been studied for shapes other than spheres as well. In \cite{Ohser1997} the case of cubic particles is considered. In \cite{Ohser1995} a variation of the problem is studied: the particles are random polyhedra, and therefore not all particles have the same shape in this setting. A system of oriented cylinders is considered in \cite{McGarrity2014}. 

The main contributions of this paper are as follows. A key insight in our instance
of the problem highlights that we can separate the shape of the particles
from their sizes in the sense that an observed area may be interpreted as the
product of two independent random variables, one related to the particle size and
the other related to the known particle shape. The density function of the shape-related random variable is explicitly known only in exceptional cases, therefore we rely on the simulation procedure proposed in \cite{vdjagt2022} that can be used to approximate it arbitrarily well. 

Using that shape-related distribution as ingredient, we design a maximum likelihood procedure to estimate the size distribution of the particles, a procedure that can be used for a large class of possible shapes. Furthermore, we show consistency of the resulting estimator and provide algorithms that can be used to compute it. Additionally, we assess the proposed estimator in a small simulation study in which we focus on convex polyhedra for the shape of the particles.   

The paper is organized as follows. In section \ref{section_preliminaries} we introduce necessary notation and definitions. In section \ref{section_stereological_equation} an integral equation is derived which describes the problem. Via this equation we obtain an identifiability result in section \ref{section_identifiability} stating that the profile area distribution uniquely determines the 3D size distribution. We define an estimator for the so-called biased size distribution in section \ref{section_estimator}. In section \ref{section_consistency} we prove consistency of this estimator. Algorithms for computing the proposed estimator are discussed in section \ref{section_algorithms}. In section \ref{section_truncation} we describe how to estimate the particle size distribution via the biased size distribution. In section \ref{section_simulations} some simulations are performed and at the end of the paper we provide some conclusions in section \ref{section_conclusion}.  

\section{Preliminaries}\label{section_preliminaries}
In this section we introduce necessary notation and collect some preliminary results which are needed in the rest of this paper. We consider a system of randomly positioned particles, and each particle is a convex body. Formally, a convex body is a convex and compact set with non-empty interior. Given a $\lambda > 0$ and a set $A \subset \RR^3$ the scalar multiplication of $\lambda$ with $A$ is defined as: $\lambda A = \{\lambda x: x \in A\}$. Let $\SO(3)$ denote the rotation group of degree 3. It contains all $3\times3$ rotation matrices, which are orthogonal matrices of determinant 1. Some definitions are necessary to precisely describe what is meant by a random plane section of a particle. The sphere in $\RR^3$ is given by: $S^2 = \{(x_1,x_2,x_3) \in \RR^3: x_1^2 + x_2^2 + x_3^2 = 1\}$. The upper hemisphere in $\RR^3$ is given by: $S_{+}^{2} = \{(x_1,x_2,x_3) \in S^2: x_3 \geq 0\}$. Let $\sigma_{2}$ denote the spherical measure on $S^{2}$, also known as the spherical Lebesgue measure on $S^{2}$. In integrals over (a subset of) $S^{2}$ the notation $\mathrm{d}\theta$ should be interpreted as $\mathrm{d}\sigma_{2}(\theta)$. A plane in $\RR^3$ may be parameterized via a unit normal vector $\theta \in S_{+}^{2}$ and its signed distance $s\in \RR$ to the origin:
\begin{equation}
    T_{\theta, s} = \{x \in \RR^3: \langle x,\theta\rangle = s\}, \label{hyperplane_eq}
\end{equation}
with $\langle \cdot,\cdot\rangle$ being the usual inner product in $\RR^3$.

We define what is meant by an Isotropic Uniformly Random (IUR) plane hitting a given convex body in $\RR^3$. The notion of IUR planes was introduced in \cite{Davy1977}, we follow the definition as in section 5.6 in \cite{Baddeley2004}.
\begin{definition}[IUR plane]\label{iur_plane_def}
An IUR plane $T$ hitting a given convex body $K \subset \RR^3$, is defined as $T = T_{\Theta, S}$ where $(\Theta, S)$ has joint probability density, $f_K : S_{+}^{2} \times \RR \to [0, \infty)$ given by:
\begin{equation*}
    f_K(\theta, s) = \begin{cases}\frac{1}{2\pi \bar{b}(K)} & \text{if } K \cap T_{\theta,s} \neq \emptyset \\
    0 & \text{otherwise,}
    \end{cases}
\end{equation*}
with $T_{\theta,s}$ as in (\ref{hyperplane_eq}) and $\bar{b}(K)$ is the mean caliper diameter, or mean width of $K$:
\[\bar{b}(K) = \frac{1}{2\pi}\int_{S_{+}^{2}} L(p_\theta(K))\mathrm{d}\theta.\]
Here, $p_\theta(K)$ represents the orthogonal projection of $K$ on the line through the origin with direction $\theta$. $L(p_\theta(K))$ is the length of this orthogonal projection, and may also be called the width of $K$ in direction $\theta$. Convexity of $K$ ensures $L(p_\theta(K))$ is the length of an interval.
\end{definition}

Loosely speaking this means that for an IUR plane through $K$, every plane which has a non-empty intersection with $K$ has equal probability of occurring. We need the following lemma, which appears as proposition 1 in \cite{Davy1977}:
\begin{lemma}\label{iur_properties_lemma}
    Suppose that $Q \subset \RR^3$ is a convex body and $K \subset Q$ is another convex body. Let $T$ be an IUR plane hitting $Q$, then:
    \begin{enumerate}
        \item Hitting probability: 
        \begin{equation*}
        \PP(T \cap K \neq \emptyset) = \frac{\bar{b}(K)}{\bar{b}(Q)}.
        \end{equation*}
        \item Conditional property: Given that $T$ hits $K$, i.e. $T \cap K \neq \emptyset$, $T$ is an IUR plane hitting $K$.
    \end{enumerate}
\end{lemma}

When a convex body $K$ is hit by an IUR plane, we obtain a section with a random area. Let $G_K$ denote the cumulative distribution function (CDF) associated with such a random area. It is sometimes referred to as cross section area distribution. The CDF $G_K$ is studied in \cite{vdjagt2022}, we collect the following properties:

\begin{theorem}\label{G_properties_theorem}
    Let $K \subset \RR^3$ be a convex body and let $T$ be an IUR plane hitting $K$. The random variable $Z := \mathrm{area}(K\cap T)$ has distribution function $G_K$. Let $G_K^S$ denote the distribution function of $\sqrt{Z}$. The following properties hold:
    \begin{enumerate}
    \item Motion invariance: $G_K$ and $G_K^S$ are invariant under translations and rotations of $K$. 
    \item Scaling of convex bodies: $G_{\lambda K}(z) = G_K\left(\frac{z}{\lambda^{2}}\right)$ for all $\lambda > 0$, $z \in \RR$.
        \item Absolute continuity: If $K$ is strictly convex or if it is a polyhedron such that each edge is parallel to at most one other edge, then $G_K$ and $G_K^S$ have a Lebesgue density.
        \item Initial monotonicity: If $G_K^S$ has Lebesgue density $g_K^S$, then $g_K^S$ is non-decreasing on $(0,\tau_K)$ for some $\tau_K > 0$.
    \end{enumerate}
\end{theorem}

Note in particular that for a large class of convex bodies, $G_K$ has a Lebesgue density. Whether $G_K$ is absolutely continuous with respect to Lebesgue measure for all convex bodies is an open problem. In section \ref{section_estimator} we define an estimator for the particle size distribution. It will then become clear why the square root transformation in Theorem \ref{G_properties_theorem} is relevant.

\section{Derivation of the stereological integral equation}\label{section_stereological_equation}
In this section we give a formal description of the model and derive a stereological integral equation. As mentioned in the introduction, stereology deals with estimating higher dimensional information from lower dimensional samples. The stereological equation in this section directly relates the distribution of the 3D particle sizes to the distribution of observed 2D section profile areas. We derive an expression for the density $f_A$ of the observed section areas. Another derivation of this density appears in chapter 16 of \cite{Santalo2004}. The derivation has two purposes, it provides an intuitive understanding of the problem and the equation is used for defining an estimator. 

Let $Q\subset \RR^3$ be the opaque convex body containing the randomly positioned particles. The intersection of $Q$ with a random plane yields a sample of observed section profile areas. For now, assume that $Q$ contains just one particle, a convex body $K_1$. Assume that $K_1$ is similar to a known convex body $K \subset \RR^3$, which we refer to as the reference particle. This means that there exists a rotation $M \in \SO(3)$, a point $x \in \RR^3$ and a scalar $\Lambda > 0$ such that $K_1 = \Lambda MK + x := \{\Lambda Mk + x: k \in K\}$. We refer to the scalar $\Lambda$ as the size of $K_1$, which is distributed according to an unknown size distribution with CDF $H$ and PDF $h$. As such the size is the scaling with respect to the reference particle, which has size 1. The mean size is denoted by:
\[\EE(\Lambda) = \int_0^\infty \lambda h(\lambda)\mathrm{d}\lambda,\]
and we assume $0 < \EE(\Lambda) < \infty$ throughout. Let $T$ be an IUR plane hitting $Q$. Let $B := \{T\cap K_1 \neq \emptyset\}$ be the event that $K_1$ is hit by $T$. By Lemma \ref{iur_properties_lemma}, the probability that $K_1$ is hit by $T$ given that it has size $\lambda$ is given by:
\begin{equation}
    \PP(B|\Lambda = \lambda) = \frac{\bar{b}(\lambda K)}{\bar{b}(Q)} = \lambda\frac{\bar{b}(K)}{\bar{b}(Q)}.\label{hitting_prob}
\end{equation}
Here, we use the fact that $\bar{b}(\lambda K) = \lambda\bar{b}(K)$ and $\bar{b}(K)$ is invariant under rotations and translations of $K$. While $\Lambda$ is drawn from $H$, the size of a particle which appears in the plane section follows a different distribution from $H$. By this we mean that $\Lambda | B$ is not distributed according to $H$. Note that the probability in (\ref{hitting_prob}) is proportional to $\lambda$, via Bayes' rule the density of $\Lambda | B$, denoted by $h^b$ is computed as:
\begin{equation*}
    h^b(\lambda) := f_{\Lambda|B}(\lambda) = \frac{\PP(B|\Lambda = \lambda)h(\lambda)}{\PP(B)} = \frac{\PP(B|\Lambda = \lambda)h(\lambda)}{\int_{0}^\infty\PP(B|\Lambda = \lambda)h(\lambda)\mathrm{d}\lambda} = \frac{\lambda h(\lambda)}{\EE(\Lambda)}.
\end{equation*}
Throughout this paper we refer to $h^b$ as the density of the length-biased size distribution associated with $h$. Let $H^b$ be the CDF corresponding to $h^b$ and note that $H$ and $H^b$ are related via:
\begin{equation}
    H^b(\lambda) = \frac{\int_0^\lambda x\mathrm{d}H(x)}{\int_0^\infty x\mathrm{d}H(x)} \text{\quad and \quad } H(\lambda) = \frac{\int_0^\lambda \frac{1}{x}\mathrm{d}H^b(x)}{\int_0^\infty \frac{1}{x}\mathrm{d}H^b(x)}, \text{\quad } \lambda \geq 0. \label{H_Hb_relationship}
\end{equation}
We refer to $H^b$ as the length-biased size distribution function, or the length-biased version of $H$. For an elaborate overview of length-biased and more generally size-biased distributions we refer to \cite{Arratia2019}. The authors also prove the following general property of length-biased distributions: if $\Lambda_b \sim H^b$ and $\Lambda \sim H$, then: $\PP(\Lambda_b \geq \lambda) \geq \PP(\Lambda \geq \lambda)$. Hence, as $H^b$ is the size distribution of the particles which appear in the plane section, this means that larger particles are more likely to appear in the cross section. 

We can now derive the distribution of an observed section area, resulting from $K_1$ being hit by the section plane. Conditional on $K_1$ being hit let $A := \text{area}(K_1 \cap T)$. By the conditional property of IUR planes in Lemma \ref{iur_properties_lemma}, given that $T$ hits $K_1$ it is an IUR plane hitting $K_1$. Therefore, if $K_1$ with size $\lambda$ appears in the section plane, its section area is distributed according to $G_{\lambda K}$. Using the rules of conditional probability we find:
\begin{align*}
F_A(a) := \PP(A \leq a | B) = \int_0^\infty \PP(A \leq a| B, \Lambda=\lambda)f_{\Lambda|B}(\lambda)\mathrm{d}\lambda = \int_0^\infty G_{\lambda K}(a)\mathrm{d}H^b(\lambda).
\end{align*}
Using point 2 of Theorem \ref{G_properties_theorem}, $F_A$ may be written as:
\begin{equation}
F_A(a) = \int_0^\infty G_{K}\left(\frac{a}{\lambda^2} \right)\mathrm{d}H^b(\lambda) = \frac{1}{\EE(\Lambda)}\int_0^\infty G_{K}\left(\frac{a}{\lambda^2}\right)\lambda\mathrm{d}H(\lambda). \label{cdf_section_areas}
\end{equation}
Suppose now that we randomly position and orient non-overlapping particles $K_1, K_2, \dots $ in $Q$, each similar to $K$. More specifically, the centers of the particles are distributed according to a homogeneous Poisson point process. As for the orientations, all orientations of the particles are equally likely and independent. The sizes of the particles $\Lambda_1,\Lambda_2,\dots$ are independent and identically distributed (iid) according to $H$. Intersecting $Q$ with an IUR plane yields an iid sample $A_1,\dots,A_n$ from $F_A$ of observed section areas, for some random $n$. Let $K$ be a convex body such that $G_K$ has a density $g_K$, recall Theorem \ref{G_properties_theorem}. Let $a_\text{max}$ be the largest possible section area of $K$, such that $g_K$ has support $(0,a_\text{max})$. Then, $F_A$ has a density given by:
\begin{equation}
    f_A(a) = \frac{1}{\EE(\Lambda)}\int_{\sqrt{\frac{a}{a_{\text{max}}}}}^\infty g_{K}\left(\frac{a}{\lambda^2}\right)\frac{1}{\lambda}\mathrm{d}H(\lambda).\label{density_stereological_equation}
\end{equation}
The stereological equation (\ref{density_stereological_equation}) directly relates the sizes of the 3D particles to the areas of the observed 2D section profiles. 

\begin{example}[Wicksell's corpuscle problem]\label{Wicksell_example}
Choose for the reference particle $K = \bar{B}(0,1)=\{x \in \RR^3: \left\Vert x \right\Vert \leq 1 \}$, the ball with radius 1. Then: $g_K(z) = 1/(2\pi\sqrt{1 - z/\pi})$, $0 < z < \pi$. We may interpret $H$ as the distribution function of the radii of the 3D balls. Note that any plane section of a ball yields a circular disc. Given $A \sim f_A$ set $A = \pi R^2$, the density of the observed circle radii satisfies: $f_R(r) = f_A(\pi r^2)2\pi r$, $r>0$. Combining this with (\ref{density_stereological_equation}) yields:
\[f_R(r) = \frac{1}{\EE(\Lambda)}\int_r^\infty \frac{1}{2\pi\sqrt{1 - \frac{r^2}{\lambda^2}}}\frac{2\pi r}{\lambda}\mathrm{d}H(\lambda) = \frac{r}{\EE(\Lambda)}\int_r^\infty \frac{1}{\sqrt{\lambda^2 - r^2}}\mathrm{d}H(\lambda),\]
which corresponds to the well-known Wicksell's integral equation \cite{Wicksell1925}.
\end{example}

\begin{remark}\label{reference_particle_interpretation_remark}
By taking an appropriate choice for the reference particle, the size distribution may be directly related to a more convenient distribution. For example, if the reference particle has diameter 1, then the size distribution corresponds to the distribution of the diameters of the particles. When choosing a reference particle with volume 1, then a particle with size $\lambda$ has volume $\lambda^3$. The volume distribution function is then given by $F_V(x) = \PP(\Lambda^3 \leq x) = H(x^{\frac{1}{3}})$. 
\end{remark}

The derived stereological equation also holds under different assumptions. The random system of particles may be defined by choosing an isotropic typical particle, and then positioning the particles using a stationary point process on $\RR^3$. This model is also known as a germ-grain model. Relevant references are sections 6.5 and 10.5 in \cite{Chiu2013}, as well as \cite{Ohser2000} and \cite{Ohser1997}. Hence, there is no need to restrict the particles to an opague body or to position the particles via a Poisson point process. 

In this setting, let $N_V$ denote the expected number of 3D particles per unit volume, which corresponds to the intensity parameter of the point process. Intersecting the system of particles with a plane, let $N_A$ denote the expected number of observed 2D section profiles per unit area. By combining the well known stereological equation (Theorem 10.1 in \cite{Chiu2013}):
\begin{equation*}
    N_A = N_V\bar{\bar{b}} \text{ \quad with \quad } \bar{\bar{b}} := \bar{b}(K)\int_0^\infty \lambda\mathrm{d}H(\lambda), 
\end{equation*}
and (\ref{cdf_section_areas}), yields:
\begin{equation}
    N_A(1 - F_A(a)) = N_V\Bar{b}(K)\int_{0}^\infty \lambda(1-G_{\lambda K}(a))\mathrm{d}H(\lambda). \label{general_stereological}
\end{equation}
A derivation of a slightly more general version of (\ref{general_stereological}) may be found in chapter 6 of \cite{Benes2004}. We specifically mention (\ref{general_stereological}) since it appears more frequently in the
literature than (\ref{density_stereological_equation}).

In order to obtain a better understanding of the problem it helps to apply a transformation. We apply a square root transformation to (\ref{cdf_section_areas}). For $A \sim F_A$, set $S = \sqrt{A}$ such that $S \sim F_S$ and $F_S(s) = F_A(s^2)$ for $s \in \RR$. As in Theorem \ref{G_properties_theorem}, let $Z \sim G_K$ then $\sqrt{Z} \sim G_K^S$ and $G_K^S(z) = G_K(z^2)$ for $z \in \RR$. Let $s \in \RR$, then the following holds:
\begin{equation}
  F_S(s) = \int_0^{\infty} G_K^S\left(\frac{s}{\lambda}\right)\mathrm{d}H^b(\lambda).\label{FS_cdf}  
\end{equation}
This expression may be recognized as the distribution function corresponding to a product of two independent random variables. This is a key insight which is made precise in the following lemma.

\begin{lemma}\label{independent_decomposition_lemma}
Consider a distribution function $H$ with length-biased version $H^b$. Suppose $Z \sim G_K$ and $\Lambda_b \sim H^b$ with $Z$ and $\Lambda_b^2$ independent. Set $A = Z\Lambda_b^2$. Then, $A \sim F_A$, and $F_A, G_K$ and $H^b$ are related via (\ref{cdf_section_areas}).
\end{lemma}

\begin{proof}
Let $X,Y,Z$ be non-negative random variables, with CDF $F_X,F_Y$ and $F_Z$ respectively. If $X=YZ$ with $Y$ and $Z$ independent, then their distribution functions are related via:
\[F_X(x) = \int_0^\infty F_Y\left(\frac{x}{z}\right)\mathrm{d}F_Z(z).\]
Comparing this with (\ref{FS_cdf}), the result is immediate.
\end{proof}

Let us provide some further intuition for Lemma \ref{independent_decomposition_lemma}. Note that point 2 of Theorem \ref{G_properties_theorem} means that for a given size $\lambda > 0$, if $Z\sim G_K$ then $Z\lambda^2 \sim G_{\lambda K}$. As the sizes of the particles in the section plane are distributed according to $H^b$, this hints towards the relationship given in Lemma \ref{independent_decomposition_lemma}. 

Therefore, there are two main considerations in this problem. First, the size distribution of particles appearing in the cross section is a length-biased version of the actual size distribution. Second, we can separate the common shape of the particles and their sizes in some sense. Taking a random size from $H^b$, and independently taking an IUR section of the reference particle yields a sample from $F_A$ via the relationship given in Lemma \ref{independent_decomposition_lemma}. 

\section{Identifiability of the particle size distribution}\label{section_identifiability}


In this section we present a general identifiability result for our model. This means that under appropriate conditions, given a known reference particle, there are no two size distributions which yield the same distribution of observed section areas. For this result we need the Mellin-Stieltjes transform, which we will also refer to as the Mellin transform. While characteristic functions appear naturally when studying sums of independent random variables, the Mellin transform is appropriate when studying products of independent random variables. We collect some properties of the Mellin transform, for details we refer to section 7.8 in \cite{Kawata1972} and \cite{Zolotarev1957}. We note that the use of the Mellin transform for this problem was already considered in \cite{Kiselak2021}. The authors obtain a slightly different expression due to the fact that an inversion formula for the density $h$ was derived and because the density $f_A$ in (\ref{density_stereological_equation}) was studied up to a normalization constant. The identfiability result in this section is new, a sufficient condition for identifiability in this context has not been derived before. 

\begin{definition}[Mellin-Stieltjes transform]
    Given a non-negative random variable $X$, with CDF $F$, the Mellin-Stieltjes transform of $X$ is defined as:
    \[\mathcal{M}_X(s) = \EE(X^{s-1}) = \int_0^\infty x^{s-1}\mathrm{d}F(x),\]
    for $s \in \CC$, whenever the integral is absolutely convergent.
\end{definition}

Note in particular, that whenever $\int x^{c-1}\mathrm{d}F(x) < \infty$ for some $c \in \RR$, then the Mellin transform exists for all $s = c + it$, $t \in \RR$. Hence, existence of the Mellin transform corresponds to the existence of moments of a distribution. Let $\st(\alpha,\beta) := \{s \in \CC: \alpha < \Re(s) < \beta\}$ denote the open strip parallel to the imaginary axis. Analogously, $\st[\alpha,\beta] := \{s \in \CC: \alpha \leq \Re(s) \leq \beta\}$ denotes the closed strip. If we find $\alpha < \beta$ such that the Mellin transform of $X$ converges absolutely on $\st[\alpha, \beta]$, then $\mathcal{M}_X$ is analytic on $\st(\alpha,\beta)$. Taking $\alpha$ as small as possible and $\beta$ as large as possible, this open strip is referred to as the strip of analyticity of $\mathcal{M}_X$. A Mellin transform uniquely determines a distribution in the following sense:

\begin{lemma}[Uniqueness of the Mellin transform]\label{Mellin_uniqueness}
Let $X \sim F_1$ and $Y \sim F_2$. Assume $\mathcal{M}_X$ and $\mathcal{M}_Y$ converge absolutely on $\st[\alpha,\beta]$, $0\leq\alpha < \beta$. If $c \in (\alpha, \beta)$ and $\mathcal{M}_X(c+it) = \mathcal{M}_Y(c+it)$ for all $t\in\RR$ then $F_1 = F_2$. 
\end{lemma}

The proof is given in Appendix \ref{appendix_proofs}. A similar statement is proven in Theorem 8 in \cite{Butzer1997} for the case that the CDF has a Lebesgue density. Finally, we recall the Mellin convolution theorem. Let $X,Y,Z$ be non-negative random variables, such that $X=YZ$ with $Y$ and $Z$ independent. For any $s \in \CC$ such that $\mathcal{M}_Y(s)$ and $\mathcal{M}_Z(s)$ are finite:
\[\mathcal{M}_X(s) = \EE\left(X^{s-1}\right) = \EE\left((YZ)^{s-1}\right) = \EE\left(Y^{s-1}\right)\EE\left(Z^{s-1}\right) = \mathcal{M}_Y(s)\mathcal{M}_Z(s).\]
Having collected these properties we now state the identifiability result.

\begin{theorem}[Identifiability]\label{identifiability_theorem}
Suppose we are given densities $f_A$, $g_K$ such that $f_A$ can be expressed as in (\ref{density_stereological_equation}) for some CDF $H$.
\begin{enumerate}
\item If $\int_0^\infty z^{-\alpha}g_K(z)\mathrm{d}z < \infty$ for some $\alpha > 0$, then there is only one distribution function $H$ on $(0,\infty)$ satisfying (\ref{density_stereological_equation}).
\item Assume $\int_0^\infty x^{1+\delta}\mathrm{d}H(x) < \infty$, for some $\delta > 0$. Then, there is only one such distribution function $H$ on $(0,\infty)$ satisfying (\ref{density_stereological_equation}).
\end{enumerate}
\end{theorem}

\begin{proof}
    We first consider statement 1 of the theorem. Let $\Lambda_b \sim H^b$, with $H^b$ as in (\ref{H_Hb_relationship}), and let $\Lambda \sim H$. Let $Z\sim g_K$ and $A\sim f_A$. We first determine on which strips the Mellin transforms of the random variables of interest are analytic. Since $\EE(\Lambda_b^{-1}) = 1/\EE(\Lambda)$ and $\EE(\Lambda_b^0) = 1$ we obtain that $\mathcal{M}_{\Lambda_b}$ is analytic on $\st(0, 1)$. Note that $1/2 < \Re(s) < 1 \iff 0 < \Re(2s-1) < 1 $. As a result: 
    \[\mathcal{M}_{\Lambda_b^2}(s) = \EE\left(\Lambda_b^{2s-2} \right) = \mathcal{M}_{\Lambda_b}(2s-1),\]
    for all $s \in \st(1/2,1)$ and $\mathcal{M}_{\Lambda_b^2}$ is analytic on $\st(1/2,1)$. Choose $\alpha > 0$ such that $\EE(Z^{-\alpha}) < \infty$. Because $g_K$ has bounded support, all non-negative moments of $Z$ exist and therefore $\mathcal{M}_{Z}$ is analytic on $\st(1-\alpha,\infty)$. By Lemma \ref{independent_decomposition_lemma} and the Mellin convolution theorem we obtain:
    \[\mathcal{M}_A(s) = \mathcal{M}_Z(s)\mathcal{M}_{\Lambda_b^2}(s),\]
    for all $s \in \st(1-\alpha,\infty)\cap\st(1/2,1) = \st(\max\{1 - \alpha,1/2\}, 1)$. Moreover, this also means that $\mathcal{M}_A$ is analytic on $\st(\max\{1 - \alpha,1/2\}, 1)$. Let $c \in (\max\{1 - \alpha,1/2\}, 1)$. Define: 
    \[L_Z := \{c+it: t\in\RR,\ \mathcal{M}_Z(c + it) \neq 0 \}, \text{ \ and \ } L := \{c+it: t \in \RR\}.\]
For all $s \in L_Z$ we find: $\mathcal{M}_A(s)/\mathcal{M}_Z(s) = \mathcal{M}_{\Lambda_{b}^2}(s)$. Define $f:L_Z \to \CC$ by $f(s) = \mathcal{M}_A(s)/\mathcal{M}_Z(s)$. Note that $f$ is analytic on $L_Z$, because $s \mapsto \mathcal{M}_{\Lambda_{b}^2}(s)$ is analytic on the line $L$. As a result there is a unique analytic continuation of $f$ to $L$. The uniqueness of this analytic continuation implies: $f(s) = \mathcal{M}_{\Lambda_{b}^2}(s)$, for all $s \in L$. Suppose $\bar{H}$ also satisfies (\ref{density_stereological_equation}), with $\bar{H}^b$ denoting its length-biased version and $\bar{\Lambda}_b \sim \bar{H}^b$. Then, following the same steps as before, we obtain: $f(s) = \mathcal{M}_{\bar{\Lambda}_{b}^2}(s)$ for all $s \in L$. By Lemma \ref{Mellin_uniqueness}, $\Lambda_{b}^2$ and $\bar{\Lambda}_{b}^2$ have the same CDF. Therefore, for all $x \in \RR$: 
\[H^b(x) = \PP(\Lambda_{b}^2 \leq x^2) = \PP(\bar{\Lambda}_b^2 \leq x^2) = \bar{H}^b(x).\] 
By (\ref{H_Hb_relationship}) this also implies $H=\bar{H}$. 

The proof of the second statement of the theorem is analogous, we simply highlight the differences. Let $\delta > 0$ be such that $\EE(\Lambda^{1+\delta}) < \infty$. Note that $\EE(\Lambda_b^{-1}) = 1/\EE(\Lambda)$ and $\EE(\Lambda_b^\delta) = \EE(\Lambda^{\delta + 1})/\EE(\Lambda)$. It then follows that $\mathcal{M}_{\Lambda_b}$ is analytic on $\st(0,1+\delta)$ and $\mathcal{M}_{\Lambda_b^2}$ is analytic on $\st(1/2,1+\delta/2)$. Clearly, $\mathcal{M}_Z$ is analytic on $\st(1,\infty)$. Hence, $\mathcal{M}_A$ is analytic on $\st(1/2,1+\delta/2) \cap \st(1,\infty) = \st(1,1+\delta/2)$. In this case we take $c \in (1,1+\delta/2)$ and the remainder of the proof is as before.
\end{proof}

We obtain as a consequence:

\begin{corollary}\label{identifiability_corollary} In the following cases the distribution function $H$ is identifiable:
\begin{enumerate}
\item The Wicksell corpuscle problem.
\item $G_K^S$ has a bounded density $g_K^S$.
\end{enumerate}
\end{corollary}

\begin{proof}
Recall the expression for $g_K$ in Example \ref{Wicksell_example}. Then:
\[\int_0^\pi z^{-\frac{1}{2}}g_K(z)\mathrm{d}z = \int_0^\pi \frac{1}{2\pi \sqrt{z}\sqrt{1 - \frac{z}{\pi}}}\mathrm{d}z = \frac{\sqrt{\pi}}{2}.\]
This integral may be computed by substituting $t=z/\pi$ and recognizing the resulting integral as an integral of a constant times the density of a Beta distribution. Hence, condition 1 of Theorem \ref{identifiability_theorem} is satisfied. If $G_K^S$ has a density $g_K^S$ with $g_K^S \leq B$, then:
\[\int_0^{a_{\text{max}}} z^{-\frac{1}{4}}g_k(z)\mathrm{d}z = \int_0^{\sqrt{a_{\text{max}}}} z^{-\frac{1}{2}}g_k^S(z)\mathrm{d}z \leq B\int_0^{\sqrt{a_{\text{max}}}} z^{-\frac{1}{2}}\mathrm{d}z = 2B(a_{\text{max}})^{\frac{1}{4}}.\]
Therefore, in this case condition 1 of Theorem \ref{identifiability_theorem} is also satisfied.
\end{proof}

Identifiability for the Wicksell problem is a classical result, in this case there is also a well-known explicit inverse relation.

\begin{remark}
The condition in Theorem \ref{identifiability_theorem}: $\int_0^\infty x^{1+\delta}\mathrm{d}H(x) < \infty$, for some $\delta > 0$, is also implied by the assumption $H(M) = 1$ for some $M > 0$. Recall the derivation of (\ref{density_stereological_equation}) in section \ref{section_stereological_equation}, a maximum size of the particles is clearly enforced by that fact that they are contained within the body $Q$. This is a typical assumption in stereological problems.
\end{remark}

Note that the proof of Theorem \ref{identifiability_theorem} also presents (a rather implicit) inversion formula for $H^b$. Assume $H^b$ is continuous. Let $c$ be as in the proof of Theorem \ref{identifiability_theorem}. Since analytic functions only have isolated zeros, $\mathcal{M}_Z(c+it) \neq 0$ for almost all $t\in \RR$. By using the Mellin inversion formula as in the proof of Lemma \ref{Mellin_uniqueness}:
    \begin{align}
        H^b\left(\sqrt{x}\right) = \PP\left(\Lambda_b^2 \leq x\right) = \lim_{T \to \infty} \frac{1}{2\pi i}\int_{c-iT}^{c+iT} -\frac{\mathcal{M}_A(s)}{\mathcal{M}_Z(s)}\frac{x^{-s+1}}{s}\mathrm{d}s, \quad x \geq 0.\label{inversion_formula}
    \end{align}
$H$ can then be retrieved via (\ref{H_Hb_relationship}). 

\section{Estimator for the length-biased particle size distribution}\label{section_estimator}
In this section we propose an estimator for the length-biased size distribution $H^b$. The proposed estimator is inspired by the approach taken in \cite{Jongbloed2001}, for Wicksell's corpuscle problem. Given the random fraction interpretation of Lemma \ref{independent_decomposition_lemma}, first estimating $H^b$ seems a natural intermediate step. We note that biased or weighted distributions frequently appear in stereology, see also section 7.5 in \cite{Ohser2000}. 

The reference particle $K$ is considered to be known and we assume that it satisfies one of the conditions in Theorem \ref{G_properties_theorem} such that $G_K^S$ has a density $g_K^S$. We stress that this also means that we consider $g_K^S$ to be known. While there are very few shapes for which an explicit expression is known for $g_K^S$, in \cite{vdjagt2022} a Monte Carlo simulation scheme is proposed which can be used to approximate such a density arbitrarily closely. To give some insight in how these densities look, see Figure \ref{g_s_estimates} for approximations of these densities for the cube, dodecahedron and tetrahedron. These approximations are obtained by computing a kernel density estimator with boundary correction, based on a sample of size $N=10^7$.

\begin{figure}[b!]
    \centering
    \makebox[\linewidth]{\makebox[\linewidth]{
    \begin{subfigure}[t]{0.5\linewidth}
        \centering
        \includegraphics[width=\linewidth]{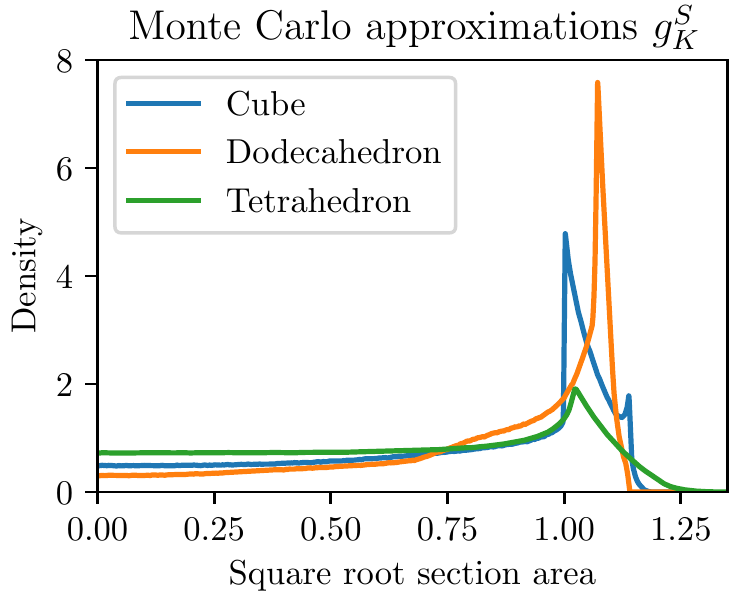}
    \end{subfigure}
    \begin{subfigure}[t]{0.5\linewidth}
        \centering
        \includegraphics[width=\linewidth]{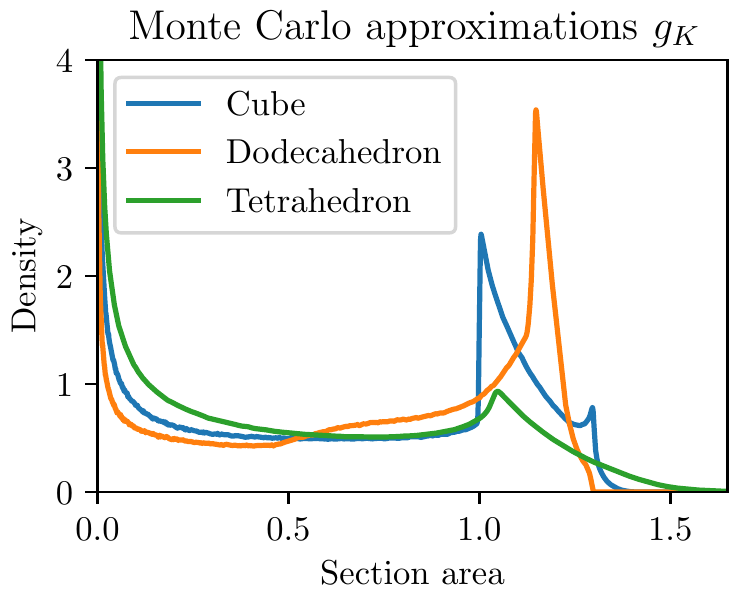}
    \end{subfigure}\hfill}}
    \caption{Left: Monte Carlo approximations of $g_K^S$, for various shapes $K$. Right: Approximations of $g_K$, obtained via $g_K(z) = g_K^S(\sqrt{z})/(2\sqrt{z})$.}
    \label{g_s_estimates}
\end{figure}

Recall the square root transformation and the resulting expression (\ref{FS_cdf}). Because $G_K^S$ has density $g_K^S$, $F_S$ has a density $f_S$ given by:
\begin{equation}
    f_S(s) = \int_0^\infty g_K^S\left(\frac{s}{\lambda}\right)\frac{1}{\lambda}\mathrm{d}H^b(\lambda).\label{sqrt_stereological_equation}
\end{equation}
Keep in mind that $g_K^S$ is supported on $(0,\sqrt{a_{\text{max}}})$, such that the lower bound of the integration region is effectively $s/\sqrt{a_{\text{max}}}$.

Suppose we have a sample of observed section areas: $A_1,\dots,A_n \overset{\mathrm{iid}}{\sim} f_A$. Let $S_i = \sqrt{A_i}$, then $S_1,\dots,S_n \overset{\mathrm{iid}}{\sim} f_S$, with $f_S$ as in (\ref{sqrt_stereological_equation}). Now, let $s_1 < s_2 <\dots < s_n$ be a realization of the order statistics of $S_1,\dots,S_n$. We use (\ref{sqrt_stereological_equation}) to implicitly define an estimator for $H^b$ via nonparametric maximum likelihood. This is achieved by considering a large class of distribution functions for $H^b$. Let $\mathcal{F}^{+}$ be the class of all distribution functions on $(0,\infty)$. Define:
\begin{equation*}
    \mathcal{F}_n^{+} = \{F \in \mathcal{F}^{+}: F \text{ is constant on } [s_{i-1}, s_{i}), i \in \{1,\dots,n\}, \text{with } F(s_0)=0\},
\end{equation*}
for some $0 < s_0 < s_1$. This means that $\mathcal{F}_n^{+}$ contains all piece-wise constant distribution functions with jump locations restricted to the set of observations, the $s_i$'s. Note that as $n \to \infty$ the set of observed $s_i$'s becomes dense in the support of $f_S$ and the class $\mathcal{F}_n^{+}$ grows to the class of all distribution functions with the same support as $f_S$. 

\begin{remark}
If $H^b(M) = 1$ for some $M>0$, then $f_S$ is supported on $(0, M\sqrt{a_{\text{max}}})$. When choosing the size of the reference particle $K$, it is important that $a_{\text{max}} \geq 1$. This is due to the choice of the sieve $\mathcal{F}_n^{+}$. Then, as $n$ tends to infinity $\mathcal{F}_n^{+}$ grows to the class of distribution functions which also contains the true CDF $H^b$, since $M\sqrt{a_{\text{max}}} \geq M$. Taking a very large $K$ means $a_{\text{max}}$ is large, such that $M\sqrt{a_{\text{max}}}$ is much larger than $M$. Then, the $s_i$'s will be quite sparse in $[0, M]$, which is also undesirable. For the sake of interpretability of $H$, recall Remark \ref{reference_particle_interpretation_remark}, we choose a $K$ with volume 1. For the shapes considered in simulations we observed $a_{\text{max}} \geq 1$. 
\end{remark}

For $H^b \in \mathcal{F}_n^{+}$ we define the (scaled by $\frac{1}{n}$) log-likelihood:
\begin{equation}
    L(H^b) := \frac{1}{n}\sum_{i=1}^n \log(f_S(s_i)) = \frac{1}{n}\sum_{i=1}^n \log\left(\int_0^\infty g_K^S\left(\frac{s_i}{\lambda}\right)\frac{1}{\lambda}\mathrm{d}H^b(\lambda) \right). \label{hb_log_likelihood}
\end{equation}
A maximum likelihood estimator (MLE) $\hat{H}_n^b$ for $H^b$ is defined as a maximizer of the log-likelihood $L$, which may be written as:
\begin{align}
    \hat{H}_n^b \in \argmax_{H^b \in \mathcal{F}_n^{+}} \frac{1}{n}\sum_{i=1}^n \log\left(\sum_{j=1}^n g_K^S\left(\frac{s_i}{s_j}\right)\frac{1}{s_j}(H^b(s_j) - H^b(s_{j-1})) \right). \label{hb_mle}
\end{align}
The following theorem shows that this estimator is well-defined, and provides a sufficient condition for uniqueness:

\begin{theorem}[Existence and uniqueness of $\hat{H}_n^b$]\label{mle_uniqueness_theorem}
A maximizer of The log-likelihood $L$ in $ \mathcal{F}_n^{+}$ always exists. The maximizer is unique if the matrix $A = (\alpha_{i,j})$, with $\alpha_{i,j} = g_K^S(s_i/s_j)/s_j$, $i,j \in \{1,\dots,n\}$, is full-rank.

\end{theorem}
\begin{proof}
For $H^b \in \mathcal{F}_n^{+}$ define: $\beta_j = H^b(s_j)$ and write $\beta = (\beta_1,\beta_2,\dots,\beta_n)^\mathsf{T}$. Consider the closed convex set:
\begin{equation}
    \mathcal{C} := \{\beta \in \RR^n: 0\leq \beta_1 \leq \beta_2 \leq \dots \leq \beta_n \leq 1\}.\label{bounded_cone_eq}
\end{equation}
The maximization problem (\ref{hb_mle}) is equivalent to maximizing $l:\mathcal{C} \to \RR \cup \{-\infty\}$ with $l$ given by:
\begin{equation}
    l(\beta) = \frac{1}{n}\sum_{i=1}^n \log\left(\sum_{j=1}^n \alpha_{i,j}(\beta_j - \beta_{j-1}) \right),\label{optimization_problem}
\end{equation}
where $\alpha_{i,j} = g_K^S(s_i/s_j)/s_j$ and $\beta_0=0$. The set $\mathcal{C}$ is closed and bounded, and therefore compact. Because of the continuity of $l$ on $\mathcal{C}$, it has a maximum. We now show that $l$ is strictly concave if and only if $A = (\alpha_{i,j})$ is full-rank. Strict concavity implies uniqueness of the maximum as well as the maximizer. Fix $\beta \in \mathcal{C}$ such that $l(\beta)>-\infty$. Let $j, k \in \{1,\dots,n\}$, computing the partial derivatives and Hessian of $l$ yields:
\begin{align}
    \frac{\partial}{\partial \beta_j}l(\beta) &= \frac{1}{n}\sum_{i=1}^n\frac{\alpha_{i,j} - \alpha_{i,j+1}}{\sum_{q=1}^n \alpha_{i,q}(\beta_q - \beta_{q-1})} \label{loglikelihood_derivatives1}\\
    \frac{\partial^2}{\partial \beta_j \partial \beta_k}l(\beta) &= -\frac{1}{n}\sum_{i=1}^n\frac{\left(\alpha_{i,j} - \alpha_{i,j+1}\right)\left(\alpha_{i,k} - \alpha_{i,k+1}\right)}{\left(\sum_{q=1}^n \alpha_{i,q}(\beta_q - \beta_{q-1})\right)^2} =:H_{j,k}(\beta).\label{loglikelihood_derivatives2}
\end{align}
Here, we have set $\alpha_{i,n+1} = 0$ for all $i \in \{1,\dots,n\}$. Since $l(\beta)>-\infty$, there are no divisions by zero in (\ref{loglikelihood_derivatives1}) and (\ref{loglikelihood_derivatives2}). Note that the following holds for $k \in \{1,\dots,n\}$:
\[\sum_{j=1}^n \alpha_{k,j}(\beta_j - \beta_{j-1}) = \sum_{j=1}^n \alpha_{k,j}\beta_j - \sum_{j = 0}^{n-1}\alpha_{k,j+1}\beta_j = \sum_{j=1}^n (\alpha_{k,j}-\alpha_{k,j+1})\beta_j. \]
Using this fact we show that the Hessian of $l$ is negative definite if and only if $A$ is full-rank. Let $\gamma \in \RR^n$ and set $\gamma_0 = 0$, then:
\begin{align}
\gamma^\mathsf{T} H(\beta)\gamma &= \sum_{j=1}^n\sum_{k=1}^n H_{j,k}(\beta)\gamma_j\gamma_k \nonumber \\
&= -\frac{1}{n}\sum_{i=1}^n\sum_{j=1}^n \sum_{k=1}^n  \frac{\left(\alpha_{i,j} - \alpha_{i,j+1}\right)\left(\alpha_{i,k} - \alpha_{i,k+1}\right)\gamma_j \gamma_k}{\left(\sum_{q=1}^n \alpha_{i,q}(\beta_q - \beta_{q-1})\right)^2} \nonumber \\
&= -\frac{1}{n}\sum_{i=1}^n \frac{\left(\sum_{j=1}^n \alpha_{i,j}(\gamma_j - \gamma_{j-1})\right)^2}{\left(\sum_{q=1}^n \alpha_{i,q}(\beta_q - \beta_{q-1})\right)^2} \nonumber 
\end{align}
Clearly, $H(\beta)$ is negative semidefinite. Note that the following holds:
\begin{equation}
\gamma^\mathsf{T} H(\beta)\gamma = 0 \iff \sum_{j=1}^n \alpha_{i,j}(\gamma_j - \gamma_{j-1}) = 0, \text{ \ for all } i \in \{1,\dots,n\}.\label{mle_uniqueness_condition}
\end{equation}
Define $x \in \RR^n$ via $x_j = \gamma_j - \gamma_{j-1}$, $j \in \{1,\dots,n\}$. Consider the matrix $A = (\alpha_{i,j})$, then the RHS of (\ref{mle_uniqueness_condition}) may be written as $Ax = 0$. Since $\gamma_0 = 0$: $\gamma=0 \iff x=0$. Therefore, the Hessian is negative definite if and only if $Ax = 0 \iff x = 0$, which corresponds to $A$ being full-rank.
\end{proof}

Recall that $g_K^S$ is supported on $(0,\sqrt{a_\text{max}})$. Suppose we choose the reference particle such that $\sqrt{a_{\text{max}}} = 1 + \epsilon$ for some small $\epsilon > 0$. Then, $g_{K}^S(1) > 0$ ensuring that the diagonal of $A$ contains positive entries. Whenever $s_i/s_j > 1+\epsilon$ for all $i>j$, $A$ is an upper triangular matrix, because all entries below the diagonal are zero. It is well-known that such matrices are of full-rank. If $\epsilon > 0$ is chosen sufficiently small, then with high probability $s_{j+1}/s_j > 1+\epsilon$ for all $j \in \{1,\dots,n\}$, such that the MLE is unique with high probability. For the sake of convenience we will refer to $\hat{H}_n^b$ as the MLE, even though we cannot always guarantee uniqueness. Note especially for the consistency result in the next section that consistency of the MLE should be interpreted as consistency of any sequence of MLE's. 

From the proof of Theorem \ref{mle_uniqueness_theorem} it is clear that $\hat{H}_n^b$ may be computed by maximizing $l$. Because $l$ is a concave function, computing $\hat{H}_n^b$ can be done efficiently as we will discuss in section \ref{section_algorithms}.

\section{Consistency of the maximum likelihood estimator}\label{section_consistency}

In this section we show that the MLE $\hat{H}_n^b$ (\ref{hb_mle}) for $H^b$ is uniformly strongly consistent. In order to prove this, we transform the problem into a deconvolution problem. Deconvolution problems have been studied before quite extensively, see for example \cite{Groeneboom2013} and \cite{Groeneboom1992}. We use some results on estimators in deconvolution problems to show consistency of the MLE. In deconvolution problems it is typical that assumptions are made on the so-called noise kernel to ensure consistency. For this problem this translates into assumptions on the density $g_K^S$.

We start by rewriting the problem of estimating $H^b$ into a deconvolution problem. Recall Lemma \ref{independent_decomposition_lemma}, for $S \sim f_S$, $\sqrt{Z} \sim g_K^S$ and $\Lambda_b \sim H^b$ we have: $S =^d \sqrt{Z} \Lambda_b$, with $\sqrt{Z}$ and $\Lambda_b$ independent. Let us now perform a log-transformation, define: $Y = \log(S)$, $\epsilon = \log(\sqrt{Z})$ and $X = \log(\Lambda_b)$. The densities of $Y$ and $\epsilon$ are related to those of $S$ and $\sqrt{Z}$ by: $f_Y(y) = f_S(e^y)e^y$, $f_\epsilon(z) = g_K^S(e^z)e^z$. The distribution function of $X$ is given by $F_X(x) = H^b(e^x)$. We then obtain:
\begin{equation*}
    Y \overset{d}{=} X + \epsilon,
\end{equation*}
with $X$ and $\epsilon$ independent. Note that $f_Y$ is the convolution of $f_{\varepsilon}$ and $F_X$:
\begin{equation}
    f_Y(y) = \int_{-\infty}^\infty f_{\varepsilon}(y - x)\mathrm{d}F_X(x) =: \left(f_{\varepsilon} * \mathrm{d}F_X\right)(y). \label{convolution_density_fy}
\end{equation} 
In this setting, $F_X$ is the distribution function of interest. We do not have direct observations from $F_X$, there is additive noise from the known distribution of $\epsilon$. Let $\mathcal{F}$ be the class of all distribution functions on $\RR$. Define:
\begin{equation*}
    \mathcal{F}_n = \{F \in \mathcal{F}: F \text{ is constant on } [y_{i-1}, y_{i}), \text{ for } i \in \{1,\dots,n\}, \text{with } F(y_0)=0\}.
\end{equation*}
The observed order statistics $s_1,\dots,s_n$ are transformed as well: $y_i = \log(s_i)$, $i \in \{0,1,\dots,n\}$. We proceed similarly as before, the log-likelihood may be written as:
\begin{equation*}
    \Tilde{L}(F) = \frac{1}{n}\sum_{i=1}^n \log\left(f_Y(y_i)\right) = \frac{1}{n}\sum_{i=1}^n \log\left(\int_{-\infty}^\infty f_{\epsilon}\left(y_i - x \right)\mathrm{d}F_X(x) \right).
\end{equation*}
A maximum likelihood estimator $\hat{F}_n$ for $F_X$ is defined as:
\begin{equation}
    \hat{F}_n \in \argmax_{F_X \in \mathcal{F}_n} \Tilde{L}(F_X). \label{mle_deconvolution}
\end{equation}
We now show that we may assume $\hat{H}_n^b(x) = \hat{F}_n(\log(x))$. The likelihoods of the two problems are related as follows. Let $F_X \in \mathcal{F}_n$, define $H^b(x) = F_X(\log(x))$ such that $H^b \in \mathcal{F}_n^{+}$. Then:
\begin{align*}
    \Tilde{L}(F_X) &= \frac{1}{n}\sum_{i=1}^n \log\left(\sum_{j=1}^n f_{\epsilon}\left(\log(s_i) - \log(s_j) \right)(F_X(\log(s_j)) - F_X(\log(s_{j-1})))\right) \\
    &= \frac{1}{n}\sum_{i=1}^n \log\left(\sum_{j=1}^n g_K^S\left(\frac{s_i}{s_j}\right)\frac{s_i}{s_j}(H^b(s_j) - H^b(s_{j-1}))\right) \\
    &= L(H^b) + \frac{1}{n}\sum_{i=1}^n \log(s_i).
\end{align*}
If we find a distribution function which provides a better likelihood in one of the problems, we immediately obtain a distribution function which provides the same improvement in likelihood in the other problem. So indeed, there exist MLE's which are related via: $\hat{H}_n^b(x) = \hat{F}_n(\log(x))$. The estimator $\hat{F}_n$ was studied in \cite{Groeneboom2013} and shown to be strongly uniformly consistent under some conditions on $f_{\varepsilon}$. This result may be used to show strong uniform consistency of $\hat{H}_n^b$, since:
\begin{equation*}
    \sup_{x > 0}|\hat{H}_n^b(x) - H^b(x)| = \sup_{x> 0}|\hat{F}_n(\log(x)) - F_X(\log(x))| = \sup_{x \in \RR}|\hat{F}_n(x) - F_X(x)|.
\end{equation*}

Let us now specify the assumptions we require for $f_{\epsilon}$. We assume it belongs to the class $\mathcal{G}$ of upper semicontinuous functions that are of bounded variation on a compact interval and monotone outside this interval. Let $V_a^b(f)$ denote the total variation of the function $f$ on the interval $[a, b]$, $a < b$. The class $\mathcal{G}$ may be written as:
\begin{align*}
\mathcal{G} &=  \bigg\{ g : \RR \to [0,\infty): g \text{ is an upper semicontinuous density such that } \exists M > 0\\
& \qquad \text{ with } V_{-M}^M (g) < \infty \text{ and } g \text{ is monotone on } (-\infty,-M] \text{ and } [M,\infty) \bigg\}.
\end{align*}
This corresponds with the following assumptions on $g_K^S$:

\begin{lemma}\label{gs_condition_lemma}
Assume that $g_K^S$ is upper semicontinuous and of bounded variation on its support. Then, the density $f_{\epsilon}:\RR \to [0,\infty)$ given by $f_{\epsilon}(z) = g_K^S(e^z)e^z$ belongs to $\mathcal{G}$.
\end{lemma}

The proof of this lemma can be found in Appendix \ref{appendix_proofs}. We now collect some lemmas to obtain a consistency result for $\hat{F}_n$. The following result can be found in \cite{Groeneboom2013}, as Corollary 1:

\begin{lemma}\label{lemma_MLE_L1_consistency}
    Let $\hat{F}_n$ be the MLE for $F_X$ defined in (\ref{mle_deconvolution}). Assume $f_{\epsilon} \in \mathcal{G}$. Set $\hat{f}_n := f_{\epsilon} * \mathrm{d}\hat{F}_n$, $f_Y := f_{\epsilon} * \mathrm{d}F_X$, then almost surely: 
    \[\lim_{n \to \infty} \Vert\hat{f}_n -f_Y\Vert_{L_1} = \lim_{n \to \infty}\int \left|\hat{f}_n(s) - f_Y(s)\right|\mathrm{d}s = 0.\]
\end{lemma}

The following lemma is a generalization of Lemma 3 in \cite{Groeneboom2013}. The proof is given in Appendix \ref{appendix_proofs}.

\begin{lemma}\label{lemma_convolution_continuity}
    Let $f_{\varepsilon}$ be a Lebesgue density on $\RR$. Let $(F_n)_{n \geq 1}$ be a sequence of distribution functions on $\RR$, converging weakly to a distribution function $F_X$. Then, for $f_n := f_{\epsilon} * \mathrm{d}F_n$ and $f_Y := f_{\epsilon} * \mathrm{d}F_X$:
    \[\lim_{n \to \infty} \Vert f_n -f_Y\Vert_{L_1} = 0.\]
\end{lemma}

We now state the following theorem, which closely follows the proof of Theorem 3 in \cite{Groeneboom2013}.

\begin{theorem}[Consistency of $\hat{F}_n$]
Let $f_{\epsilon} \in \mathcal{G}$. Assume that the deconvolution problem with this $f_\epsilon$ is identifiable. Then, with probability one, $\lim_{n\to \infty} \hat{F}_n(x) = F_X(x)$ for each $x$ where $F_X$ is continuous. If $F_X$ is continuous then with probability one:
\[\lim_{n\to \infty} \left\Vert\hat{F}_n - F_X\right\Vert_{\infty} = 0.\]
\end{theorem}

\begin{proof}
Let $(\Omega, \mathcal{A}, \PP)$ be a probability space supporting a sequence $Y_1,Y_2,\dots$ of iid random variables, distributed according to $f_Y$ as in (\ref{convolution_density_fy}). Set $\hat{f}_n := f_{\epsilon} * \mathrm{d}\hat{F}_n$. By Lemma \ref{lemma_MLE_L1_consistency} we know there exists a set $\Omega_0 \in \mathcal{A}$ with $\PP(\Omega_0) = 1$ such that for all $\omega \in \Omega_0$ we have $\Vert\hat{f}_n(\cdot,\omega) - f_Y(\cdot)\Vert_{L_1} \to 0$ as $n \to \infty$. Fix $\omega \in \Omega_0$ and choose an arbitrary subsequence $(n_l)_{l \geq 1} \subset (n)_{n \geq 1}$. By Helly's selection principle there exists a further subsequence $(n_k)_{k \geq 1} \subset (n_l)_{l \geq 1}$ such that $\hat{F}_{n_k}(\cdot, \omega)$ converges weakly to a distribution function $F$. By Lemma \ref{lemma_convolution_continuity} this implies that $\hat{f}_{n_k}$ converges to $f_{\epsilon} * \mathrm{d}F$ in $L_1$. Because the whole sequence $\hat{f}_n$ converges to $f_{\epsilon} * \mathrm{d}F_X$ in $L_1$ this implies $F=F_X$ by identifiability  of the deconvolution problem. Therefore, every subsequence of MLE's contains a further subsequence converging weakly to $F_X$. This implies weak convergence of the whole sequence to $F_X$. Finally, the uniform result follows from the monotonicity of all distribution functions in the sequence and $F_X$, and continuity of $F_X$.
\end{proof}

Turning to a consistency result for $\hat{H}_n^b$, we need to make sure that $g_K^S$ satisfies the conditions in Lemma \ref{gs_condition_lemma}. If $g_k^S$ satisfies these conditions, its boundedness implies the problem is identifiable by Corollary \ref{identifiability_corollary}. Note that identifiability in the original problem implies identifiability in the corresponding deconvolution problem. 

\begin{corollary}[Consistency of $\hat{H}_n^b$]
Assume $g_K^S$ is upper semicontinuous and of bounded variation on its support. Then, with probability one, $\lim_{n \to \infty} \hat{H}_n^b(\lambda) = H^b(\lambda)$ for each $\lambda$ where $H^b$ is continuous. If $H$ is continuous, then so is $H^b$, and with probability one:
\[\lim_{n\to \infty} \left\Vert\hat{H}_n^b - H^b\right\Vert_{\infty} = 0.\]
\end{corollary}

\section{Algorithms}\label{section_algorithms}
In this section we describe some algorithms for computing the maximum likelihood estimator $\hat{H}_n^b$ (\ref{hb_mle}). Since a distribution function in $\mathcal{F}_n^{+}$ is discrete, it may be described by a probability vector. Let $\mathcal{P}_n$ be the class of probability vectors in $\RR^n$:
\begin{equation*}
    \mathcal{P}_n = \left\{(p_1,\dots,p_n) \in \RR^n: \sum_{i=1}^n p_i = 1 \text{ and }  p_i \geq 0 \text{ for all } i \in \{1,\dots,n\} \right\}.
\end{equation*}
A distribution function $H^b \in \mathcal{F}_n^{+}$ may be associated with the probability vector $p \in \mathcal{P}_n$ defined as $p_j = H^b(s_j) - H^b(s_{j-1})$ (recall: $H^b(s_0) = 0$). We can switch between probability vectors and distribution functions via:
\begin{equation}
    H^b(s_j) = \sum_{i=1}^j p_i \text{, \quad and \quad} p_j = H^b(s_j) - H^b(s_{j-1}).\label{switch_probability_vector}
\end{equation}

\subsection{Expectation Maximization (EM)}
The EM algorithm was first thoroughly studied in \cite{Dempster1977}. While it is typically used in parametric settings, it may also be used for non-parametric estimation. It is especially appealing due to its ease of implementation and its interpretation for incomplete data models. For the application of EM to our problem, we follow the description of EM in \cite{Wellner1997}. The authors describe the EM algorithm for problems similar to the one we are facing. The class of problems they consider is the following. 

Suppose we aim to estimate a distribution function $F$. We cannot directly observe a sample $X$ from $F$. Instead, we observe $Y = T(X,C)$, with $X \sim F$ and $C$ some random variable independent of $X$. Clearly, given Lemma \ref{independent_decomposition_lemma} the problem of estimating $H^b$ belongs to this class of problems with $T(x,c)=xc$ and $C \sim g_K^S$. Suppose we have an initial estimate $H_0^{b}\in \mathcal{F}_n^{+}$ of the CDF $H^b$. Let $p^{(0)}$ be the associated probability vector as in (\ref{switch_probability_vector}). Let $X_1,\dots,X_n \sim H^b$. In \cite{Wellner1997} it is shown that in their general context the EM algorithm yields the following update rule:
\begin{equation}
    p_j^{(k+1)} = \frac{1}{n}\sum_{i=1}^n \PP_{p^{(k)}}\left(X_i = s_j |s_1,\dots,s_n \right).\label{em_update_rule_1}
\end{equation}
We use the notation $\PP_{p}$ to indicate the probability measure associated with the probability vector $p$. For our 'random fraction' setting, we use Bayes' rule to obtain:
\begin{equation*}
    \PP_{p^{(k)}}\left(X_i = s_j |s_1,\dots,s_n \right) = \frac{g_K^S\left(\frac{s_i}{s_j}\right)\frac{1}{s_j}p_j^{(k)}}{\sum_{q=1}^n g_K^S\left(\frac{s_i}{s_q}\right)\frac{1}{s_q}p_q^{(k)}}.
\end{equation*}
Plugging this into (\ref{em_update_rule_1}) yields:
\begin{equation}
    p_j^{(k+1)} = \frac{1}{n}\sum_{i=1}^n \frac{\alpha_{i, j}}{\sum_{q=1}^n \alpha_{i, q} p_q^{(k)}}p_j^{(k)}, \text{ \quad with: } \alpha_{i,j} = g_K^S\left(\frac{s_i}{s_j}\right)\frac{1}{s_j}.\label{em_update_rule_final}
\end{equation}
When terminating the EM algorithm after an appropriate number of iterations we obtain $\hat{H}_n^b$ from $p^{(k)}$ via (\ref{switch_probability_vector}). The EM algorithm may for example be terminated when successive iterations do not meaningfully change the log-likelihood anymore. We do not provide a specific stopping criterion for EM, since we do not use it directly. We only use it in hybrid form with the Iterative Convex Minorant algorithm (ICM) which is described in the next section. For the ICM algorithm and the hybrid ICM-EM algorithm we do provide explicit termination conditions. 

\subsection{Iterative Convex Minorant (ICM)}
The ICM algorithm was first introduced in \cite{Groeneboom1992}. The version of ICM we discuss is described in \cite{Jongbloed1998} and is sometimes called the modified ICM algorithm. This modification of ICM ensures convergence under fairly general conditions. The algorithm is designed to minimize a convex function $\phi$ over the closed convex cone:
\begin{equation*}
    \mathcal{C}_{+} := \{\beta \in \RR^n: 0\leq \beta_1 \leq \beta_2 \leq \dots \leq \beta_n\}.
\end{equation*}
Recall equation (\ref{optimization_problem}), computing the estimator $\hat{H}_n^b$ is equivalent to solving the following optimization problem:
\begin{equation}
    \hat{\beta} \in \argmax_{\beta \in \mathcal{C}} l(\beta) = \argmax_{\beta \in \mathcal{C}} \frac{1}{n}\sum_{i=1}^n \log\left(\sum_{j=1}^n \alpha_{i,j}(\beta_j - \beta_{j-1}) \right). \label{optimization problem_icm}
\end{equation}
With $\beta_j = H^b(s_j)$ and $\hat{\beta} = (\hat{H}_n^b(s_1),\dots,\hat{H}_n^b(s_n))$.  From (\ref{optimization problem_icm}) it is clear that for any $\beta \in \mathcal{C}_{+}$ with $\beta_n < 1$, the likelihood can be increased by setting $\beta_n = 1$, since $\alpha_{i,j} \geq 0$. Hence, we may incorporate the constraint $\beta_n = 1$ instead of $\beta_n \leq 1$. We achieve this via a Lagrange multiplier. Define the convex function $\phi:\mathcal{C}_{+} \to \RR \cup \{\infty\}$ as:
\begin{equation*}
    \phi(\beta) = -l(\beta) + \beta_n = -\frac{1}{n}\sum_{i=1}^n \log\left(\sum_{j=1}^n \alpha_{i,j}(\beta_j - \beta_{j-1}) \right) + \beta_n.
\end{equation*}
Hereby we have incorporated the constraint, with a Lagrange multiplier equal to one. Also, the problem is now written as a convex minimization problem since $\hat{\beta} \in \argmin_{\beta \in \mathcal{C}_{+}}\phi(\beta)$. Therefore, the ICM algorithm may be used to compute the MLE. Suppose we have some initial estimate $\beta^{(0)}$. The idea of ICM is to locally approximate $\phi$ with the following quadratic form in iteration $k$:
\begin{equation}
     \phi_{(k)}(\beta) =\left(\beta - \beta^{(k)} + W\left(\beta^{(k)}\right)^{-1}\nabla \phi\left(\beta^{(k)}\right)\right)^\mathsf{T} W\left(\beta^{(k)}\right)\left(\beta - \beta^{(k)} + W\left(\beta^{(k)}\right)^{-1}\nabla \phi\left(\beta^{(k)}\right)\right).\label{quadratic_form_icm}
\end{equation}
The notation $\nabla \phi$ is used for the gradient of $\phi$, the vector of partial derivatives of $\phi$. The matrix $W$ is a diagonal matrix, its diagonal is often chosen equal to the diagonal of the Hessian matrix of $\phi$:
\begin{equation*}
    W\left(\beta^{(k)}\right) = \mathrm{diag}\left(\frac{\partial^2}{\partial \beta_j^2}\phi\left(\beta^{(k)}\right)\right)
\end{equation*}

In the ICM algorithm, $\phi_{(k)}$ is minimized over $\mathcal{C}_{+}$ instead of $\phi$ to obtain a candidate $\beta$ for $\beta^{(k+1)}$. If this candidate $\beta$ sufficiently decreases $\phi$ it is accepted, and we set $\beta^{(k+1)} = \beta$. Otherwise, a line-search is performed to obtain $\beta^{(k+1)}$, which is then given by a convex combination of $\beta$ and $\beta^{(k)}$. A precise description of the algorithm is given in Appendix \ref{appendix_algorithms}. We remark that minimizing (\ref{quadratic_form_icm}) is equivalent to computing the weighted least-squares estimator of a monotone regression function. This can be done efficiently, for more details see \cite{Jongbloed1998}. The partial derivatives of $\phi$ are related to those of $l$ as in (\ref{loglikelihood_derivatives1}) and (\ref{loglikelihood_derivatives2}), via:
\[\frac{\partial}{\partial \beta_j}\phi(\beta) = -\frac{\partial}{\partial \beta_j}l(\beta) + \mathds{1}{\{j=n\}} \quad \text{and} \quad \frac{\partial^2}{\partial \beta_j^2}\phi(\beta) = -\frac{\partial^2}{\partial \beta_j^2}l(\beta).\]
For ICM we use the following stopping criterion, stop whenever:
\begin{equation}
    \max_{j \in \{1,\dots,n\}}\left|\beta_j^{(k)} - \beta_j^{(k-1)}\right| < \varepsilon, \label{stopping_criterion}
\end{equation}
for 10 successive iterations. In simulations we set $\varepsilon = 10^{-4}$. The interpretation of this criterion is that we stop whenever the largest change in probability mass is below $\varepsilon$ for 10 successive iterations. We note that this criterion could be inappropriate if ICM approaches the optimum very slowly. In simulations (section \ref{section_simulations}) this was not an issue.

\subsection{Hybrid ICM-EM}
In \cite{Wellner1997} it was proposed to combine ICM and EM into a hybrid algorithm. The idea is that a single iteration of this hybrid algorithm consists of first performing one iteration of the ICM algorithm followed by one iteration of the EM algorithm. The ICM algorithm appears somewhat slow initially, if it is started far from the MLE, whereas close to the optimal value it converges quickly. On the contrary, the EM algorithm seems quicker at the start but has trouble converging when close to the optimum. Moreover, when performing an EM step after an ICM step, ICM ensures that many of the $p_j$'s are zero. From (\ref{em_update_rule_final}) we see that EM will never set such a $p_j$ to a positive value, hence EM only needs to operate in a lower dimensional space. In practice it seems that the hybrid ICM-EM algorithm inherits the strengths of both algorithms and is quicker than both ICM and EM. This was for example observed in simulations in \cite{Jongbloed2001} and \cite{Wellner1997}. As with ICM, the same termination condition (\ref{stopping_criterion}) is used.

\section{Regularization of the maximum likelihood estimator}\label{section_truncation}
In this section we describe how the MLE $\hat{H}_n^b$ may be used to estimate $H$, the distribution function of interest. At first glance it seems reasonable to plug in $\hat{H}_n^b$ for $H^b$ in equation (\ref{H_Hb_relationship}). Unfortunately, simulations indicate that this yields a poor estimate of $H$. In section \ref{section_simulations} we describe in detail how simulations are performed. For now, Figure \ref{naive_h_estimation_figure} shows the result of a single simulation run. This simulation corresponds to the case where each particle is a dodecahedron, $n=1000$, and $H$ corresponds to a standard exponential distribution. For this $H$, $H^b$ corresponds to a gamma distribution. Figure \ref{naive_h_estimation_figure} (a) shows that $\hat{H}_n^b$ closely resembles $H^b$. Meanwhile, in Figure \ref{naive_h_estimation_figure} (b) we observe that plugging in $\hat{H}_n^b$ for $H^b$ in equation (\ref{H_Hb_relationship}) yields a poor estimate of $H$. This is due to the influence of the behavior of $\hat{H}_n^b$ near zero.

\begin{figure}[b!]
    \centering
    \makebox[\linewidth]{\makebox[\linewidth]{
    \begin{subfigure}[t]{0.5\linewidth}
        \centering
        \includegraphics[width=\linewidth]{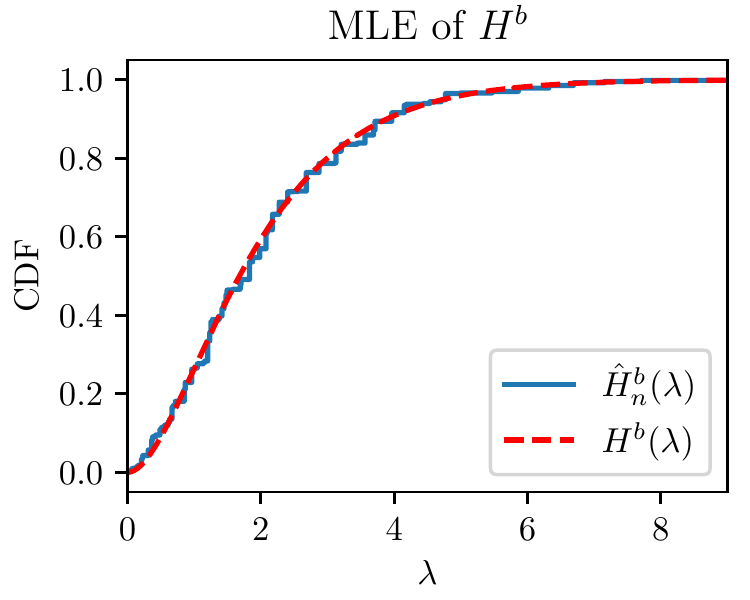}
    \end{subfigure}
    \begin{subfigure}[t]{0.5\linewidth}
        \centering
        \includegraphics[width=\linewidth]{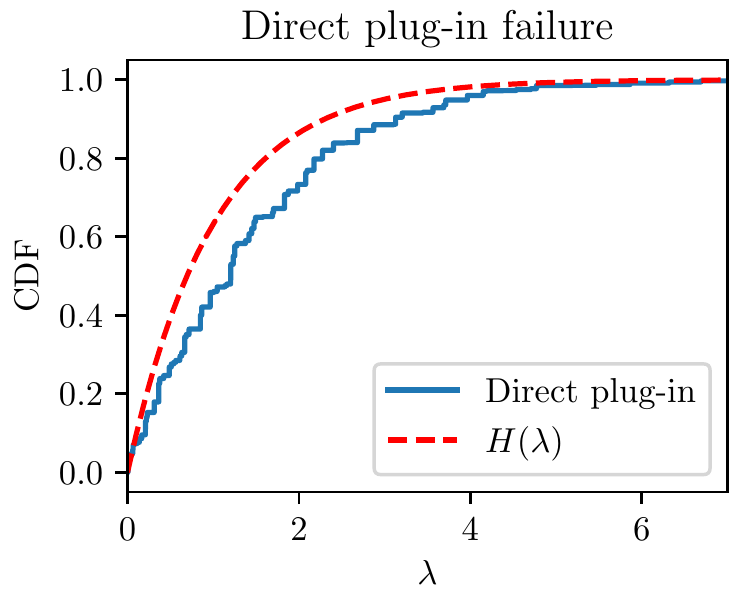}
    \end{subfigure}\hfill}}
    \caption{Left: MLE of $H^b$. Right: Direct plug-in estimate of $H$.}
    \label{naive_h_estimation_figure}
\end{figure}

We propose a regularization technique to resolve this issue. Let $t_n > 0$, truncating $\hat{H}_n^b$ at $t_n$ yields:
\begin{align*}
    \hat{H}_n^b(\lambda ; t_n) := \begin{cases} \frac{\hat{H}_n^b(\lambda) - \hat{H}_n^b(t_n)}{1 - \hat{H}_n^b(t_n)} & \text{ if } \lambda \geq t_n\\
0 & \text{ otherwise}\end{cases}.
\end{align*}
Plugging this truncated version of $\hat{H}_n^b$ into (\ref{H_Hb_relationship}) we obtain:
\begin{equation}
    \hat{H}_n(\lambda; t_n) := \frac{\int_0^\lambda \frac{1}{x}\mathrm{d}\hat{H}_n^b(x; t_n)}{\int_0^\infty \frac{1}{x}\mathrm{d}\hat{H}_n^b(x; t_n)} = \begin{cases}\frac{\int_{t_n}^\lambda \frac{1}{x}\mathrm{d}\hat{H}_n^b(x)}{\int_{t_n}^\infty \frac{1}{x}\mathrm{d}\hat{H}_n^b(x)} & \text{ if } \lambda \geq t_n \\
    0 & \text{ if } 0 \leq \lambda < t_n \end{cases}. \label{truncation_estimator}
\end{equation}
Therefore, we introduce a new parameter $t_n$, which we refer to as the truncation parameter. In the following lemma we show that for an appropriate choice of the truncation parameter $t_n$, a sequence of approximating CDFs converging to $H^b$ may be de-biased to obtain a close approximation of $H$.

\begin{lemma}\label{consistent_truncation_theorem}
Let $H$ be a continuous CDF on $(0,\infty)$, with finite first moment and length-biased version $H^b$. Let $(t_n)_{n \geq 1}$, $t_n > 0$ be a sequence such that $\lim_{n\to \infty}t_n = 0$. Let $(H_n^b)_{n \geq 1}$ be a sequence of CDFs. Assume $H_n^b$ converges uniformly to $H^b$ with rate at least $t_n$, that is: $||H_n^b - H^b||_{\infty} = o(t_n)$. Define:
\begin{equation*}
    H_n(\lambda) = \begin{cases}\frac{\int_{t_n}^\lambda \frac{1}{x}\mathrm{d}H_n^b(x)}{\int_{t_n}^\infty \frac{1}{x}\mathrm{d}H_n^b(x)} & \text{ if } \lambda \geq t_n \\
    0 & \text{ if } 0 \leq \lambda < t_n \end{cases},
\end{equation*}
then: $\lim_{n \to \infty} ||H_n - H||_{\infty} = 0$.
\end{lemma}

The proof is given in Appendix \ref{appendix_proofs}. Lemma \ref{consistent_truncation_theorem} shows that truncation is a viable approach for consistent estimation of $H$. Note that in Lemma \ref{consistent_truncation_theorem}, we may take $t_n = \sqrt{\Vert H_n^b - H^b \Vert_{\infty}}$. In practice the result cannot directly be applied to $\hat{H}_n^b$ since the quantity $\Vert \hat{H}_n^b - H^b \Vert_{\infty}$ is unknown. We propose a rule of thumb for $t_n$. Let $s \in \RR$ and define:
\begin{align}
    \hat{F}_n^S(s;t) &:= \int_0^\infty G_K^S\left(\frac{s}{\lambda}\right)\mathrm{d}\hat{H}_n^b(\lambda ; t) \label{induced_distribution_FS} \\
    \bar{F}_n^S(s) &:= \frac{1}{n}\sum_{i=1}^n \mathds{1}{\{s_i \leq s\}}.\nonumber
\end{align}
Note that $\hat{F}_n^S(\cdot;t)$ is the distribution function of observed square root section areas induced by the biased size distribution $\hat{H}_n^b(\cdot ; t)$. That is, if $\hat{H}_n^b(\cdot ; t)$ is the true biased size distribution, then $\hat{F}_n^S(\cdot;t)$ is the corresponding distribution of observed square root section areas. We propose the following choice for $t_n$:
\begin{equation}
\hat{t}_n := \argmin_{t \in \{s_1,\dots,s_n\}} \int_0^\infty |\hat{F}_n^S(s;t) - \bar{F}_n^S(s)|\mathrm{d}s. \label{truncation_parameter_estimator}
\end{equation}
Hence, $\hat{t}_n$ minimizes the $L^1$-distance between the CDF of the observed square root section areas, induced by the estimated (biased) size distribution, and the empirical CDF of observed square root section areas. We minimize over $\{s_1,\dots,s_n\}$ for computational convenience.

\section{Simulations}\label{section_simulations}
In the previous sections we have introduced the MLE $\hat{H}_n^b$, and shown that under reasonable assumptions it is a consistent estimator of $H^b$. Also, a regularization technique was introduced to consistently estimate the size distribution function $H$ using the MLE. In this section some simulations results are presented to assess the performance of these estimators for $H^b$ and $H$. The code used for the simulations may be found at \url{https://github.com/thomasvdj/pysizeunfolder}. Using this code the simulation and estimation procedure can be carried out in principle for any choice of convex polyhedron for the reference particle $K$.

Let us start by describing how to generate an iid sample of observed section areas, for a given $H$ and a chosen reference particle $K$. Lemma \ref{independent_decomposition_lemma} shows that it is sufficient to draw $Z \sim G_K$ and independently draw $\Lambda_b \sim H^b$, followed by setting $A:=Z\Lambda_b^2$. $A$ may be considered a random section area, and repeating these steps $n$ times yields an iid sample $A_1,\dots,A_n$ distributed according to $f_A$. Taking the square root yields a sample of observed square root section areas. A sampling scheme for generating IUR planes through $K$ is described in \cite{Davy1977}, see \cite{vdjagt2022} for drawing from $G_K$. Finally, we consider some well-known parametric distributions for $H$, for these choices $H^b$ corresponds to some other well-known parametric distribution. Hence, drawing from $H^b$ is straightforward. The following choices for $H$ are considered, with the corresponding $H^b$:
\begin{enumerate}
\item \emph{Exponential distribution}: For $H$ we consider a standard exponential distribution, such that $H^b$ corresponds to a gamma distribution.
\[H(\lambda) = 1 - e^{-\lambda}, \text{ \ and \ } H^b(\lambda) = 1 - (\lambda + 1)e^{-\lambda}, \ \lambda \geq 0.\]
\item \emph{Lognormal distribution}: For $H$ we consider a lognormal distribution with parameters $\mu$ and $\sigma$. For this $H$, $H^b$ corresponds to a lognormal distribution with parameters $\mu + \sigma^2$ and $\sigma$. We set $\mu=2$, $\sigma = 1/2$.
\[H(\lambda) = \Phi\left(\frac{\log(\lambda) - \mu}{\sigma} \right) , \text{ \ and \ } H^b(\lambda) = \Phi\left(\frac{\log(\lambda) - \mu - \sigma^2}{\sigma} \right), \ \lambda > 0.\]
Here, $\Phi$ denotes the CDF of a standard normal distribution.
\end{enumerate}

\begin{figure}[b!]
\centering
    \makebox[\linewidth]{\makebox[\linewidth]{
    \begin{subfigure}[t]{0.5\linewidth}
        \centering
        \includegraphics[width=\linewidth]{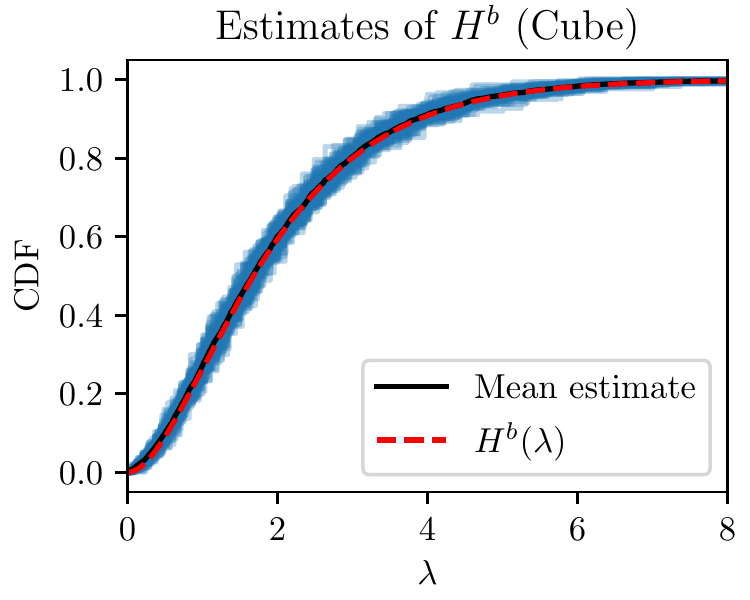}
    \end{subfigure}
    \begin{subfigure}[t]{0.5\linewidth}
        \centering
        \includegraphics[width=\linewidth]{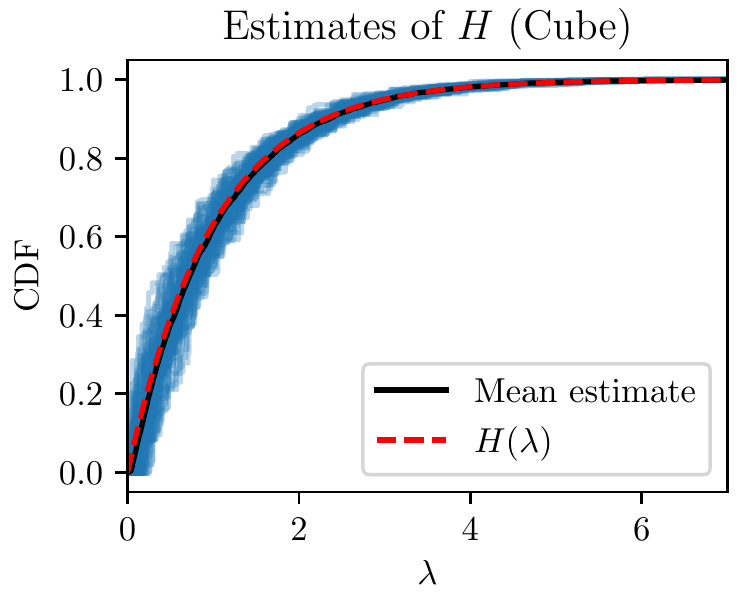}
    \end{subfigure}\hfill}}
    \makebox[\linewidth]{\makebox[\linewidth]{
    \begin{subfigure}[t]{0.5\linewidth}
        \centering
        \includegraphics[width=\linewidth]{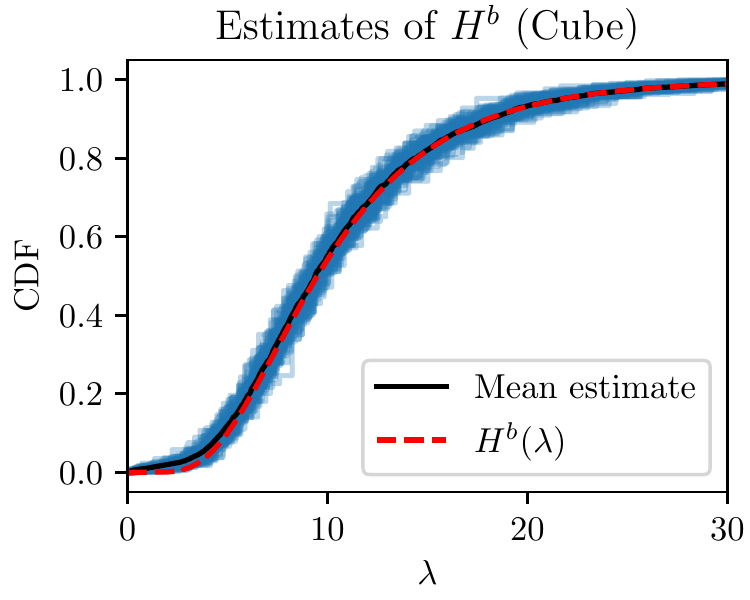}
    \end{subfigure}
    \begin{subfigure}[t]{0.5\linewidth}
        \centering
       \includegraphics[width=\linewidth]{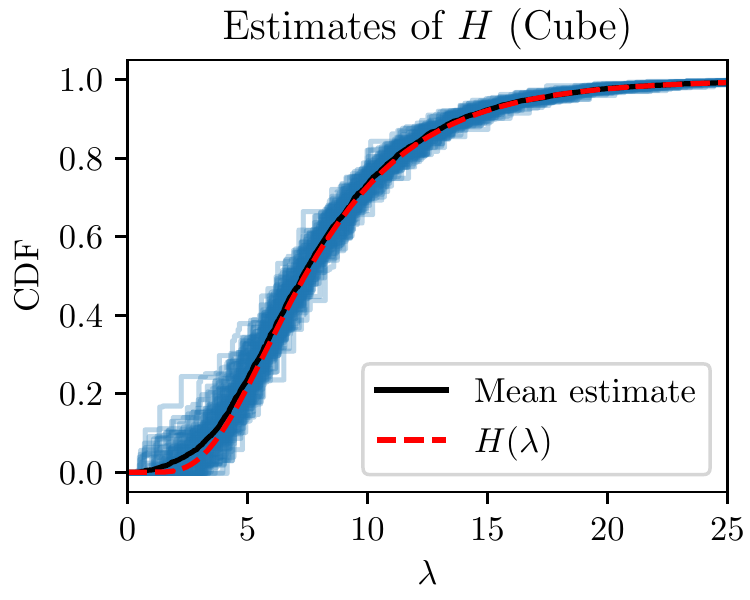}
    \end{subfigure}\hfill}}
    \caption{Simulation results for the cube, $n=1000$. Top left: $H^b$ is a gamma distribution. Top right: $H$ is an exponential distribution. Bottom left: $H$ is a lognormal distribution. Bottom right: $H$ is a lognormal distribution.}
    \label{cube_simulations_figure}
\end{figure}

\begin{table}[b!]
\centering
\caption{Simulation results for the dodecahedron.}\label{dodecahedron_table}
\begin{tabular}
    {
  r
  l
  S[table-format=1.4, table-column-width=2cm]
>{{{\lp}}} 
S[round-precision=2, table-format = 1.3,table-space-text-pre=\lp]
@{,\,} 
S[round-precision=2, table-format = 1.3,table-space-text-post=\rp]
<{{{\rp}}} 
    S[table-format=1.4]
>{{{\lp}}} 
S[round-precision=2,table-format = 1.3,table-space-text-pre=\lp]
@{,\,} 
S[round-precision=2,table-format = 1.3,table-space-text-post=\rp]
<{{{\rp}}} 
@{}l@{}
}
\toprule
\multicolumn{2}{c}{} & \multicolumn{3}{c}{$\Vert \hat{H}_n^b - H^b \Vert_{\infty}$} & \multicolumn{3}{c}{$\Vert \hat{H}_n - H \Vert_{\infty}$} &\\
\cmidrule(r){1-1}\cmidrule(lr){2-2}\cmidrule(lr){3-5}\cmidrule(l){6-9} 
  \multicolumn{1}{c}{$n$} &        \multicolumn{1}{c}{$H$} & {mean error} &    \multicolumn{2}{c}{(2.5\%, 97.5\%)} & {mean error} &    \multicolumn{2}{c}{(2.5\%, 97.5\%)} & \\
\cmidrule(r){1-1}\cmidrule(lr){2-2}\cmidrule(lr){3-5}\cmidrule(l){6-9} 
 1000 & Exponential &   0.057671 & 0.038322 & 0.076175 &   0.118394 & 0.064229 & 0.202112 &   \\
 1000 &   Lognormal &   0.065665 & 0.044518 & 0.104524 &   0.092446 & 0.057303 & 0.180478 &   \\
 2000 & Exponential &   0.045230 & 0.031512 & 0.062793 &   0.097153 & 0.054488 & 0.166696 &   \\
 2000 &   Lognormal &   0.052757 & 0.033663 & 0.081283 &   0.078280 & 0.044436 & 0.125059 &   \\
 5000 & Exponential &   0.031843 & 0.023648 & 0.044521 &   0.068698 & 0.035735 & 0.136639 &   \\
 5000 &   Lognormal &   0.037992 & 0.027410 & 0.054203 &   0.058649 & 0.037004 & 0.097180 &   \\
10000 & Exponential &   0.026044 & 0.019095 & 0.035397 &   0.057817 & 0.028961 & 0.117645 &   \\
10000 &   Lognormal &   0.029804 & 0.022833 & 0.039797 &   0.047790 & 0.027651 & 0.085952 &   \\
\bottomrule
\end{tabular}

\vspace{5pt}
\caption{Simulation results for the cube.}\label{cube_table}
\begin{tabular}
    {
  r
  l
  S[table-format=1.4, table-column-width=2cm]
>{{{\lp}}} 
S[round-precision=2, table-format = 1.3,table-space-text-pre=\lp]
@{,\,} 
S[round-precision=2, table-format = 1.3,table-space-text-post=\rp]
<{{{\rp}}} 
    S[table-format=1.4]
>{{{\lp}}} 
S[round-precision=2,table-format = 1.3,table-space-text-pre=\lp]
@{,\,} 
S[round-precision=2,table-format = 1.3,table-space-text-post=\rp]
<{{{\rp}}} 
@{}l@{}
}
\toprule
\multicolumn{2}{c}{} & \multicolumn{3}{c}{$\Vert \hat{H}_n^b - H^b \Vert_{\infty}$} & \multicolumn{3}{c}{$\Vert \hat{H}_n - H \Vert_{\infty}$} &\\
\cmidrule(r){1-1}\cmidrule(lr){2-2}\cmidrule(lr){3-5}\cmidrule(l){6-9} 
  \multicolumn{1}{c}{$n$} &        \multicolumn{1}{c}{$H$} & {mean error} &    \multicolumn{2}{c}{(2.5\%, 97.5\%)} & {mean error} &    \multicolumn{2}{c}{(2.5\%, 97.5\%)} & \\
\cmidrule(r){1-1}\cmidrule(lr){2-2}\cmidrule(lr){3-5}\cmidrule(l){6-9} 
 1000 & Exponential &   0.064703 & 0.044717 & 0.093717 &   0.133794 & 0.072821 & 0.229255 &   \\
 1000 &   Lognormal &   0.079351 & 0.055051 & 0.119813 &   0.107424 & 0.061791 & 0.166210 &   \\
 2000 & Exponential &   0.050867 & 0.035528 & 0.069692 &   0.106668 & 0.058667 & 0.193202 &   \\
 2000 &   Lognormal &   0.062959 & 0.043395 & 0.087236 &   0.091057 & 0.050023 & 0.172565 &   \\
 5000 & Exponential &   0.039365 & 0.028866 & 0.052752 &   0.078655 & 0.042695 & 0.131728 &   \\
 5000 &   Lognormal &   0.045957 & 0.032615 & 0.059284 &   0.067227 & 0.039993 & 0.108839 &   \\
10000 & Exponential &   0.030813 & 0.022231 & 0.041903 &   0.061987 & 0.036163 & 0.095589 &   \\
10000 &   Lognormal &   0.036842 & 0.027549 & 0.046776 &   0.054416 & 0.031381 & 0.091156 &   \\
\bottomrule
\end{tabular}

\vspace{5pt}
\caption{Simulation results for the tetrahedron.}\label{tetrahedron_table}

\begin{tabular}
    {
  r
  l
  S[table-format=1.4, table-column-width=2cm]
>{{{\lp}}} 
S[round-precision=2, table-format = 1.3,table-space-text-pre=\lp]
@{,\,} 
S[round-precision=2, table-format = 1.3,table-space-text-post=\rp]
<{{{\rp}}} 
    S[table-format=1.4]
>{{{\lp}}} 
S[round-precision=2,table-format = 1.3,table-space-text-pre=\lp]
@{,\,} 
S[round-precision=2,table-format = 1.3,table-space-text-post=\rp]
<{{{\rp}}} 
@{}l@{}
}
\toprule
\multicolumn{2}{c}{} & \multicolumn{3}{c}{$\Vert \hat{H}_n^b - H^b \Vert_{\infty}$} & \multicolumn{3}{c}{$\Vert \hat{H}_n - H \Vert_{\infty}$} &\\
\cmidrule(r){1-1}\cmidrule(lr){2-2}\cmidrule(lr){3-5}\cmidrule(l){6-9} 
  \multicolumn{1}{c}{$n$} &        \multicolumn{1}{c}{$H$} & {mean error} &    \multicolumn{2}{c}{(2.5\%, 97.5\%)} & {mean error} &    \multicolumn{2}{c}{(2.5\%, 97.5\%)} & \\
\cmidrule(r){1-1}\cmidrule(lr){2-2}\cmidrule(lr){3-5}\cmidrule(l){6-9} 
 1000 & Exponential &   0.094815 & 0.061885 & 0.146928 &   0.197051 & 0.089983 & 0.387729 &   \\
 1000 &   Lognormal &   0.109597 & 0.078199 & 0.152344 &   0.163102 & 0.091326 & 0.297084 &   \\
 2000 & Exponential &   0.079238 & 0.057953 & 0.103738 &   0.152570 & 0.084589 & 0.282576 &   \\
 2000 &   Lognormal &   0.093009 & 0.069480 & 0.126959 &   0.134373 & 0.082102 & 0.240345 &   \\
 5000 & Exponential &   0.060153 & 0.046412 & 0.080346 &   0.119844 & 0.061118 & 0.245369 &   \\
 5000 &   Lognormal &   0.076100 & 0.058399 & 0.093358 &   0.099715 & 0.067513 & 0.145027 &   \\
10000 & Exponential &   0.051411 & 0.038026 & 0.063890 &   0.101110 & 0.053744 & 0.204433 &   \\
10000 &   Lognormal &   0.064268 & 0.047599 & 0.082674 &   0.080482 & 0.059274 & 0.111282 &   \\
\bottomrule
\end{tabular}
\end{table}

For the simulations we consider the following shapes for the particles: the dodecahedron, cube and tetrahedron. As for the specific choice of the reference particle $K$, each of the shapes are scaled such that they have volume 1. Note that for the dodecahedron each edge is parallel to exactly one other edge and the tetrahedron does not have any parallel edges. Therefore these shapes are such that $G_K$ has a Lebesgue density by Theorem \ref{G_properties_theorem}. This is not the case for the cube. For the cube we could consider a perturbed cube by slightly tilting each of its faces, the resulting shape does not have any parallel edges. Note that we rely on Monte-Carlo approximations of $g_K^S$, we refer to \cite{vdjagt2022} for further discussion on why it is reasonable to apply this density approximation procedure to the cube, even though it is not covered by Theorem \ref{G_properties_theorem}.

Now that we covered the simulation of iid samples we discuss the computation of estimators. For a given choice of $n$, $H$ and shape for the particles we generate a sample of $n$ observed (square root) section areas. The MLE $\hat{H}_n^b$ is computed using the hybrid ICM-EM algorithm. The computation of the MLE requires that we can evaluate $g_K^S$ in given points. As mentioned before, there is typically no explicit expression for $g_K^S$ and we use the Monte Carlo simulation scheme described in \cite{vdjagt2022} for approximating $g_K^S$ (recall Figure \ref{g_s_estimates}). For estimating $H$ we compute $\hat{H}_n(\cdot, \hat{t}_n)$ as in (\ref{truncation_estimator}), with $\hat{t}_n$ as in (\ref{truncation_parameter_estimator}). Throughout this section we refer to this estimator simply as $\hat{H}_n$. Note that for the computation of $\hat{t}_n$ we require $G_K^S$, which is also not explicitly known. Hence, similarly to $g_K^S$ we use a Monte-Carlo approximation of $G_K^S$. In this case we use an empirical distribution function based on the same sample used for approximating $g_K^S$.

We perform repeated simulations as follows. For various choices of $n$ we generate a sample of $n$ observed section areas. This is repeated 100 times for each choice of $n$, $H$ and shape for the particles. Simulation results for the cube are shown in Figure \ref{cube_simulations_figure}. These results correspond to $n=1000$. Each of the blue lines corresponds to one of the 100 estimates, each estimate based on a different sample of size $n=1000$. The black line is the point-wise average of all estimates. Further simulation results for the other shapes are summarized in Tables \ref{dodecahedron_table}, \ref{cube_table} and \ref{tetrahedron_table}. We quantify the error of the estimate as the supremum distance between the true $H^b$ and $\hat{H}_n^b$, and similarly for the error of the estimates of $H$. The mean error is then the mean taken over the 100 resulting errors of the estimates. For these 100 resulting errors the $2.5\%$ and $97.5\%$ quantiles are also shown.

Let us discuss the content of Tables \ref{dodecahedron_table}, \ref{cube_table} and \ref{tetrahedron_table}. As expected, as $n$ increases the average error decreases, for all chosen shapes and size distributions, both for the estimates of $H$ and $H^b$. Comparing the average supremum error for a fixed $n$, and a fixed size distribution, it is clear that the errors are smallest for the dodecahedron, followed by the cube and finally the average error is largest for the tetrahedron. This is the case for both the average errors for estimating $H$ as well as $H^b$. Note that estimating $H$ instead of $H^b$ increases the supremum error, and the corresponding mean supremum errors are also larger. These larger errors are also evident in Figure \ref{cube_simulations_figure}. We note that for some practical applications an estimate of $H^b$ may be sufficient.

\begin{table}[t!]
\centering
\caption{Algorithms mean run-times and mean number of iterations.}\label{algorithm_times_table}
\begin{tabular}
    {
  r
  S[table-format=3.2, table-column-width=1.2cm]
  S[table-format=4.0, table-column-width=1.6cm]
  S[table-format=2.3, table-column-width=1.2cm]
  S[table-format=2.1, table-column-width=1.6cm]
}
\toprule
\multicolumn{1}{c}{} & \multicolumn{2}{c}{ICM} & \multicolumn{2}{c}{ICM-EM} \\
\cmidrule(r){1-1}\cmidrule(lr){2-3}\cmidrule(l){4-5} 
  \multicolumn{1}{c}{$n$} & {time $(s)$} & {\# iterations} & {time $(s)$} & {\# iterations} \\
\cmidrule(r){1-1}\cmidrule(lr){2-3}\cmidrule(l){4-5}  
1000 &      3.988466 &              313.7 &      0.510773 &               26.4    \\
2000 &     27.589594 &              502.0 &      2.423374 &               30.7    \\
5000 &    415.241904 &              784.3 &     25.808362 &               62.5    \\
\bottomrule
\end{tabular}
\end{table}

Finally, we briefly touch upon computational efficiency of the algorithms for computing $\hat{H}_n^b$. We take for the shape of the particles the dodecahedron and for $H$ the previously introduced lognormal distribution. In Table \ref{algorithm_times_table} the average run-times and iteration counts of the ICM and ICM-EM algorithms are shown, averaged over 10 simulation runs. The EM algorithm is not included in the table, in simulations it was several orders of magnitude slower than the other algorithms. Clearly, ICM-EM is considerably faster than ICM. 

\section{Concluding remarks}\label{section_conclusion}
In this paper we have studied a generalization of the classical Wicksell corpuscle problem, considering an arbitrary convex shape for the particles instead of spheres. In particular, for the problem of estimating the CDF $H$ of the particle size distribution an identifiability result is derived. We also obtain an inversion formula via the Mellin transform. A nonparametric maximum likelihood estimator is proposed for the biased size distribution $H^b$ and it is proven to be uniformly strongly consistent. Moreover, this estimator can be computed efficiently in practice. In a simulation study the proposed estimators for $H^b$ and $H$ perform well for various choices of particle shapes and particle size distributions. 

\appendix
\section{Appendix: proofs}\label{appendix_proofs}

\begin{proof}[Proof of Lemma \ref{Mellin_uniqueness}]
The result follows almost immediately from theorem 7.8.2. in \cite{Kawata1972}, which is a Mellin inversion theorem. Suppose $X \sim F_1$ and $Y \sim F_2$. By assumption $\mathcal{M}_X$ and $\mathcal{M}_Y$ are analytic on $\st(\alpha,\beta)$, $0 \leq \alpha < \beta$. Let $c \in (\alpha,\beta)$, and assume $\mathcal{M}_X(c+it) = \mathcal{M}_Y(c+it)$ for all $t\in\RR$. Let $x>0$, by theorem 7.8.2. from \cite{Kawata1972} we obtain:
\begin{align*}
\bar{F}_1(x) &:= \frac{1}{2}(F_1(x+) + F_1(x-)) = \lim_{T \to \infty} \frac{1}{2\pi i}\int_{c-iT}^{c+iT} -\mathcal{M}_X(s)\frac{x^{-s+1}}{s}\mathrm{d}s\\
\bar{F}_2(x) &:= \frac{1}{2}(F_2(x+) + F_2(x-)) = \lim_{T \to \infty} \frac{1}{2\pi i}\int_{c-iT}^{c+iT} -\mathcal{M}_Y(s)\frac{x^{-s+1}}{s}\mathrm{d}s.
\end{align*}
Here: $F(x+) := \lim_{h \downarrow 0} F(x+h)$ and $F(x-) := \lim_{h \uparrow 0} F(x+h)$. Note that for a continuity point $x$ of $F_1$, $\bar{F}_1(x) = F_1(x)$. Because CDFs are right continuous we obtain: $F_1(x) = \bar{F}_1(x+)$ and $F_2(x) = \bar{F}_2(x+)$. Hence, it is sufficient to show $\bar{F}_1 =\bar{F}_2$. Because $\mathcal{M}_X(c+it) = \mathcal{M}_Y(c+it)$ for all $t\in\RR$:
\[\bar{F}_1(x) - \bar{F}_2(x) = \lim_{T \to \infty} \frac{1}{2\pi i}\int_{c-iT}^{c+iT} -(\mathcal{M}_X(s)-\mathcal{M}_Y(s))\frac{x^{-s+1}}{s}\mathrm{d}s = 0,\]
which finishes the proof.
\end{proof}

\begin{proof}[Proof of Lemma \ref{gs_condition_lemma}]
$f_{\epsilon}$ is upper semicontinuous, as it is given by a product, and a composition of an upper semicontinuous function and a continuous function. By Theorem \ref{G_properties_theorem}, $g_K^S$ is non-decreasing on $(0, \tau_K)$ for some $0 < \tau_K \leq \sqrt{a_{\text{max}}}$. Choose $M > \sqrt{a_{\text{max}}}$ large enough such that $e^{-M} < \tau_K$. It now immediately follows that $f_{\epsilon}(z) = 0$ for $z \in [M,\infty)$ and $f_{\epsilon}$ is monotonically increasing on $(-\infty,-M]$. It remains to show that $f_{\epsilon}$ is of bounded variation on $[-M, M]$.
Let $-M < z_0 < z_1 < \dots < z_m < M$ be an arbitrary partition of $[-M,M]$. Then it follows:
\begin{align}
    \sum_{i=1}^m |f_{\epsilon}(z_i) - f_\epsilon(z_{i-1})| &= \sum_{i=1}^m |g_K^S(e^{z_i})e^{z_i} - g_K^S(e^{z_{i}})e^{z_{i-1}} + g_K^S(e^{z_{i}})e^{z_{i-1}} - g_K^S(e^{z_{i-1}})e^{z_{i-1}}| \nonumber \\
    &\leq ||g_K^S||_{\infty} \sum_{i=1}^m |e^{z_i} - e^{z_{i-1}}| + e^M \sum_{i=1}^m|g_S(e^{z_{i}}) - g_S(e^{z_{i-1}})| \label{telescoping_sum_gs_condition}\\
    &\leq ||g_K^S||_{\infty} e^M +  e^M V_0^{\sqrt{a_{\text{max}}}}\left(g_K^S\right) < \infty \nonumber.
\end{align}
Note that the first sum in (\ref{telescoping_sum_gs_condition}) telescopes. In the final step we use the fact that $g_K^S$ is bounded and is of bounded variation on its support. Because the above computation holds for arbitrary partitions of $[-M, M]$ we find: $V_{-M}^M(f_{\epsilon}) < \infty$, which finishes the proof.
\end{proof}

\begin{proof}[Proof of Lemma \ref{lemma_convolution_continuity}]
    Because $f_{\epsilon} \geq 0$ is a Lebesgue density, for every $m \in \NN$ there exists a bounded continuous probability density function $f_{\epsilon}^m$ such that $||f_{\epsilon}^m - f_{\epsilon}||_{L^1} \leq 1/m$ (see Lemma \ref{continuous_density_approximation} in Appendix \ref{appendix_proofs}). Let $m \in \NN$, then:
    \begin{align}
        ||f_n -f_Y||_{L_1} &= \int \left|\int f_{\epsilon}(z-x) - f_{\epsilon}^m(z-x) + f_{\epsilon}^m(z-x)\mathrm{d}(F_n -F_X)(x) \right|\mathrm{d}z \nonumber \\
        \begin{split}
            &\leq \int \left|\int f_{\epsilon}(z-x) - f_{\epsilon}^m(z-x)\mathrm{d}(F_n -F_X)(x) \right|\mathrm{d}z \\ 
            &\qquad + \int \left|\int f_{\epsilon}^m(z-x)\mathrm{d}(F_n -F_X)(x) \right|\mathrm{d}z .\label{convolution_continuity_proof}
        \end{split}
    \end{align}
Via the triangle inequality and Fubini, the first term in (\ref{convolution_continuity_proof}) is bounded by:
\begin{align}
        &\phantom{\leq} \int \left|\int f_{\epsilon}(z-x) - f_{\epsilon}^m(z-x)\mathrm{d}F_n(x) \right|\mathrm{d}z + \int \left|\int f_{\epsilon}(z-x) - f_{\epsilon}^m(z-x)\mathrm{d}F_X(x) \right|\mathrm{d}z \nonumber\\
        &\leq \int \int \left| f_{\epsilon}(z-x) - f_{\epsilon}^m(z-x)\right| \mathrm{d}z\mathrm{d}F_n(x) +  \int \int \left| f_{\epsilon}(z-x) - f_{\epsilon}^m(z-x)\right| \mathrm{d}z\mathrm{d}F_X(x)\nonumber\\ 
        &\leq 2||f_{\epsilon}^m - f_{\epsilon}||_{L^1} \leq \frac{2}{m}. \nonumber
\end{align}
    The second term in (\ref{convolution_continuity_proof}) may be written as:
    \[\int \left|\int f_{\epsilon}^m(z-x)\mathrm{d}(F_n -F_X)(x) \right|\mathrm{d}z = || \varphi_{n,m} - \varphi_{m}||_{L^1}, \]
    with $\varphi_{n,m}$ and $\varphi_{m}$ defined as:
    \begin{align*}
        \varphi_{n,m}(z) = \int f_{\epsilon}^m(z-x)\mathrm{d}F_n(x),
\text{\quad and \quad} \varphi_{m}(z) = \int f_{\epsilon}^m(z-x)\mathrm{d}F_X(x).  \end{align*}
Because $f_{\epsilon}^m$ is a probability density, so are $\varphi_{m}$ and $\varphi_{n,m}$ for all $n \in \NN$. By the continuity of $f_{\epsilon}^m$ and the weak convergence of $F_n$ to $F_X$ we obtain that $\varphi_{n,m}$ converges pointwise to $\varphi_{m}$ as $n \to \infty$. By Scheff\'e's Theorem pointwise convergence of probability densities to another probability density implies that these densities also converge in $L^1$. Combining all results yields:
\[\lim_{n \to \infty}||f_n -f_Y||_{L_1} \leq \lim_{n \to \infty}\frac{2}{m} + || \varphi_{n,m} - \varphi_{m}||_{L^1}= \frac{2}{m}. \]
Letting $m \to \infty$ we obtain the desired result.
\end{proof}

\begin{lemma}\label{continuous_density_approximation}
Let $f$ be a Lebesgue density on $\RR$, for every $\varepsilon > 0$ there exists a bounded continuous probability density function $g$ such that $\left\Vert g - f\right\Vert_{L^1} < \varepsilon$.
\end{lemma}
\begin{proof}
Recall that the space of compactly supported continuous functions is dense in $L^1$. For $n \in \NN$ choose a continuous, compactly supported and non-negative function $g_n$ such that $\left\Vert g_n - f\right\Vert_{L^1} \leq 1 / (n + 1)$. By the reverse triangle inequality:
\begin{equation}
\left|\left\Vert g_n\right\Vert_{L^1} - 1 \right| = \left|\left\Vert g_n\right\Vert_{L^1} - \left\Vert f\right\Vert_{L^1}\right| \leq \left\Vert g_n - f\right\Vert_{L^1} \leq \frac{1}{n+1}. \label{continuous_approximation_proof}
\end{equation}
Define: $\tilde{g}_n = g_n/\left\Vert g_n\right\Vert_{L^1}$. Note that by (\ref{continuous_approximation_proof}), $\left\Vert g_n\right\Vert_{L^1} > 0$. Hence, $\tilde{g}_n$ is a bounded and continuous probability density function. Combining all results:
\begin{align*}
    \left\Vert \tilde{g}_n - f \right\Vert_{L^1} &= \frac{1}{\left\Vert g_n\right\Vert_{L^1}}\int \left|g_n(x) - f(x) + f(x) - \left\Vert g_n\right\Vert_{L^1}f(x) \right|\mathrm{d}x\\
    &\leq \frac{1}{\left\Vert g_n\right\Vert_{L^1}} \left(\left\Vert g_n - f\right\Vert_{L^1} + \left| \left\Vert g_n\right\Vert_{L^1} - 1\right|\cdot \left\Vert f\right\Vert_{L^1} \right) \\
    &\leq \frac{\frac{1}{n+1} + \frac{1}{n+1}}{1 -\frac{1}{n+1}} = \frac{2}{n}.
\end{align*}
Because this holds for all $n \in \NN$ we obtain the desired result.
\end{proof}

\begin{proof}[Proof of Lemma \ref{consistent_truncation_theorem}]
    We first note the following:
    \begin{equation}
    \sup_{0 \leq \lambda < t_n}|H_n(\lambda) - H(\lambda)| = \sup_{0 \leq \lambda < t_n} H(\lambda) = H(t_n). \label{truncation_proof_part1}
\end{equation}
Let us now assume $\lambda \geq t_n$. By definition:
\begin{equation}
    H(\lambda) - H_n(\lambda) = \frac{\int_0^\lambda \frac{1}{x}\mathrm{d}H^b(x)\int_{t_n}^\infty \frac{1}{x}\mathrm{d}H_n^b(x) - \int_{t_n}^\lambda \frac{1}{x}\mathrm{d}H_n^b(x)\int_0^\infty \frac{1}{x}\mathrm{d}H^b(x)}{\int_0^\infty \frac{1}{x}\mathrm{d}H^b(x)\int_{t_n}^\infty \frac{1}{x}\mathrm{d}H_n^b(x)}.\label{truncation_estimator_difference}
\end{equation}
The numerator of (\ref{truncation_estimator_difference}) may be written as:
\begin{align}
\begin{split}
    &\phantom{=} \int_{t_n}^\infty \frac{1}{x}\mathrm{d}H_n^b(x)\left(\int_0^\lambda \frac{1}{x}\mathrm{d}H^b(x) - \int_{t_n}^\lambda \frac{1}{x}\mathrm{d}H_n^b(x)\right)\\ &\phantom{=}\quad  - \int_{t_n}^\lambda \frac{1}{x}\mathrm{d}H_n^b(x) \left(\int_0^\infty \frac{1}{x}\mathrm{d}H^b(x) - \int_{t_n}^\infty \frac{1}{x}\mathrm{d}H_n^b(x) \right) \nonumber
\end{split}\\
\begin{split}
        &= \int_{t_n}^\infty \frac{1}{x}\mathrm{d}H_n^b(x)\left(\int_{t_n}^\lambda \frac{1}{x}\mathrm{d}(H^b - H_n^b)(x) + \int_0^{t_n}\frac{1}{x}\mathrm{d}H^b(x) \right)\\
        &\phantom{=}\quad - \int_{t_n}^\lambda \frac{1}{x}\mathrm{d}H_n^b(x)\left(\int_{t_n}^\infty \frac{1}{x}\mathrm{d}(H^b - H_n^b)(x) + \int_0^{t_n}\frac{1}{x}\mathrm{d}H^b(x) \right). \label{numerator_truncation_proof}
\end{split}
\end{align}
Recall: $\EE(\Lambda) = \int_0^\infty \lambda \mathrm{d}H(\lambda) = 1/\int_0^\infty (1/x) \mathrm{d}H^b(x)$. Plugging (\ref{numerator_truncation_proof}) back into (\ref{truncation_estimator_difference}) yields:
\begin{align*}
    H(\lambda) - H_n(\lambda) &= \EE(\Lambda)\left(\int_{t_n}^\lambda \frac{1}{x}\mathrm{d}(H^b - H_n^b)(x)\right) + H(t_n) \\
    &\quad - \EE(\Lambda)H_n(\lambda)\left(\int_{t_n}^\infty \frac{1}{x}\mathrm{d}(H^b - H_n^b)(x)\right) - H_n(\lambda)H(t_n).
\end{align*}
Therefore, we obtain the following bound:
\begin{equation}
\sup_{\lambda \geq t_n}\left|H(\lambda) - H_n(\lambda)\right| \leq 2 \EE(\Lambda)\sup_{\lambda \geq t_n} \left|\int_{t_n}^\lambda \frac{1}{x}\mathrm{d}(H^b - H_n^b)(x) \right| + H(t_n).    \label{truncation_proof_bound}
\end{equation}
The integral in (\ref{truncation_proof_bound}) may be computed via integration by parts:
\begin{align}
    \sup_{\lambda \geq t_n} &\left|\int_{t_n}^\lambda \frac{1}{x}\mathrm{d}(H^b - H_n^b)(x) \right|  = \nonumber\\
    &= \sup_{\lambda \geq t_n} \left|\frac{H_n^b(\lambda) - H^b(\lambda)}{\lambda} - \frac{H_n^b(t_n) - H^b(t_n)}{t_n} - \int_{t_n}^\lambda \left(H_n^b(x) - H^b(x)\right)\mathrm{d}\frac{1}{x}\right|\nonumber \\
    &\leq 2 \frac{\sup_{\lambda \geq t_n} \left|H_n^b(\lambda) - H^b(\lambda) \right|}{t_n} + \sup_{\lambda \geq t_n}\left|H_n^b(\lambda) - H^b(\lambda) \right|\cdot \left|\int_{t_n}^\lambda \mathrm{d}\frac{1}{x}\right|\nonumber \\
    &\leq 3 \frac{\Vert H_n^b - H^b\Vert_{\infty}}{t_n}. \label{truncation_proof_integral_bound}
\end{align}
Note that the bound in (\ref{truncation_proof_bound}) is greater than $H(t_n)$, by (\ref{truncation_proof_part1}) this means that the bound also holds when taking the supremum over $\lambda \geq 0$ instead. Combining (\ref{truncation_proof_bound}) and (\ref{truncation_proof_integral_bound}) we finally obtain:
\begin{equation}
    \Vert H_n - H\Vert_{\infty} \leq 6\EE(\Lambda)\frac{\Vert H_n^b - H^b\Vert_{\infty}}{t_n} + H(t_n).\label{truncation_proof_final_bound}
\end{equation}
Letting $n$ go to infinity, $H(t_n)$ converges to zero by the continuity of $H$. Using this and the fact that $(H_n^b)_{n \geq 1}$ converges uniformly to $H^b$ with rate $t_n$ (by assumption) the RHS of (\ref{truncation_proof_final_bound}) converges to zero.
\end{proof}

\section{Appendix: pseudo-code of algorithms}\label{appendix_algorithms}

\begin{algorithm}[H]
\caption{Expectation Maximization (EM)}
\begin{algorithmic}[1]
\Require Observed order statistics: $s_1 < s_2 < \dots < s_n$.
\Ensure The MLE $\hat{H}_n^b$. 
\State $k := 0$
\State $p^{(0)} := (\frac{1}{n},\frac{1}{n},\dots,\frac{1}{n}) \in \mathcal{P}_n$
\While{Stopping criterion is not met}
\State$p_j^{(k+1)} := \frac{1}{n}\sum_{i=1}^n \frac{\alpha_{i, j}}{\sum_{q=1}^n \alpha_{i, q} p_q^{(k)}}p_j^{(k)}, \text{ \quad with: } \alpha_{i,j} = g_K^S\left(\frac{s_i}{s_j}\right)\frac{1}{s_j}$
\State $k := k + 1$
\EndWhile
\State $\hat{H}_n^b(s_j) := \sum_{i=1}^j p_i^{(k)}$ for $j \in \{1,\dots,n\}$.
\State
\Return $\hat{H}_n^b$
\end{algorithmic}
\label{EMalgorithm}
\end{algorithm}

\begin{algorithm}[H]
\caption{Iterative Convex Minorant (ICM)}
\begin{algorithmic}[1]
\Require A convex function $\phi:\mathcal{C}_{+} \to \RR \cup \{\infty\}$. $\beta^{(0)}$ with $\phi(\beta^{(0)}) < \infty$.
\Ensure A minimizer of $\phi$. 
\State $k := 0$
\State $\beta^{(0)} := (\frac{1}{n},\frac{2}{n},\dots,\frac{n}{n}) \in \mathcal{C}$
\While{Stopping criterion is not met}
\State $\beta := \argmin_{y \in \mathcal{C}_{+}}\phi_{(k)}(y)$ \Comment{With $\phi_{(k)}$ as in (\ref{quadratic_form_icm})}
\If{$\phi(\beta) < \phi(\beta^{(k)}) + \epsilon \nabla\phi(\beta^{(k)})^\intercal(\beta - \beta^{(k)})$}
\State $\beta^{(k+1)} := \beta$
\Else
\State $\lambda :=1$, $s:=\frac{1}{2}$, $z:=\beta$.
\While{$\phi(z) < \phi(\beta^{(k)}) + (1-\epsilon)\nabla\phi(\beta^{(k)})^\intercal (z - \beta^{(k)})$ (I) {\bf or} \\ $\phi(z) > \phi(\beta^{(k)}) + \epsilon\nabla\phi(\beta^{(k)})^\intercal (z - \beta^{(k)})$ (II)}
\IfThen{(I)}{$\lambda := \lambda +s$} 
\EndIf
\IfThen{(II)}{$\lambda := \lambda -s$}
\EndIf
\State $z:=\beta^{(k)} + \lambda(\beta - \beta^{(k)})$
\State $s := \frac{s}{2}$ 
\EndWhile
\State $\beta^{(k+1)} := z$
\EndIf
\State $k := k + 1$
\EndWhile
\State
\Return {$\beta^{(k)}$}
\end{algorithmic}
\label{ICMalgorithm}
\end{algorithm}

\section*{Acknowledgements}
We thank Kees Bos, Jilt Sietsma and Karo Sedighiani for fruitful discussions. 

\bibliographystyle{ieeetr}
\bibliography{export,misc}

\end{document}